\documentclass[10pt,acmsmall,screen,authorversion,noacm]{acmart}
\usepackage{balance}
\usepackage{xcolor}
\usepackage{url}
\usepackage{graphicx}
\usepackage{listings}
\usepackage{enumitem}
\usepackage{amsmath}
\usepackage{pifont}
\usepackage{bigstrut}
\usepackage{mathpartir}
\usepackage[altpo]{backnaur}
\usepackage{nccmath}
\usepackage{sepnum}
\usepackage{caption}
\usepackage{threeparttable}
\usepackage[linesnumbered,ruled,noend]{algorithm2e}
\usepackage{natbib}
\usepackage{microtype}
\usepackage{wrapfig}
\usepackage{subcaption}
\usepackage{multirow}
\usepackage{cleveref}
\usepackage{multicol}
\usepackage[framemethod=TikZ]{mdframed}
\usepackage{arydshln}

\newcommand{\cmark}{\ding{51}}%
\newcommand{\xmark}{\ding{55}}%

\SetKw{Break}{break}
\SetKw{Continue}{continue}
\newcounter{example}

\newcounter{boxedd}

\definecolor{eggshell}{rgb}{0.94, 0.92, 0.84}
\definecolor{floralwhite}{rgb}{1.0, 0.98, 0.94}

\mdfdefinestyle{example}{%
	frametitle=\mbox{},
	frametitlerulewidth=0.00pt,
	frametitlebelowskip=0.00pt,
	frametitlerule=false,
    frametitlefont={\footnotesize},
    backgroundcolor=floralwhite!1,
	innerleftmargin=0.5pt,
	innerrightmargin=3pt,
    linecolor=eggshell,
    topline=false,
    bottomline=true,
    leftline=false,rightline=true,
    outerlinewidth=0.5pt,
    roundcorner=2.5mm,
    skipabove=\baselineskip,
    skipbelow=\baselineskip,
}
\newmdenv[%
    style=example,
    settings={\global\refstepcounter{example}},
    frametitlefont={\bfseries\small Example bug:\quad},
]{example}

\SetCommentSty{mycommfont}

\PassOptionsToPackage{hyphens}{url}\usepackage{hyperref}

\setcopyright{rightsretained}
\setlist[itemize]{noitemsep, leftmargin=1em}
\setlist[enumerate]{noitemsep, leftmargin=1em}
\newcommand{\nnum}[1]{\sepnum{.}{,}{}{#1}}

\newcommand{\prim}[1]{\text{{\sffamily{\tt #1}}}}
\newcommand{\primsem}[1]{\text{{\sffamily{\scriptsize {\tt #1}}}}}

\setlist[itemize]{noitemsep, leftmargin=1em}
\setlist[enumerate]{noitemsep, leftmargin=1em}

\definecolor{dkred}{RGB}{87,10,10}
\definecolor{burntorange}{rgb}{0.8, 0.33, 0.0}
\definecolor{OrangeRed}{rgb}{1.0, 0.27, 0.0}
\definecolor{ForestGreen}{rgb}{0.0, 0.27, 0.13}

\definecolor{groovyblue}{HTML}{0000A0}
\definecolor{groovygreen}{HTML}{008000}
\definecolor{darkgray}{rgb}{.4,.4,.4}

\lstdefinelanguage{Groovy}[]{Java}{
  keywordstyle=\color{groovyblue}\bfseries,
  stringstyle=\color{groovygreen}\ttfamily,
  keywords=[3]{each, findAll, groupBy, collect, inject, eachWithIndex},
  morekeywords={def, as, in, use},
  moredelim=[is][\textcolor{darkgray}]{\%\%}{\%\%},
  moredelim=[il][\textcolor{darkgray}]{§§}
}

\lstdefinelanguage{Kotlin}{
  comment=[l]{//},
  emph={delegate, filter, first, firstOrNull,  lazy, map, mapNotNull, println, return@},
  emphstyle={\color{OrangeRed}},
  identifierstyle=\color{black},
  keywords={abstract, actual, as, as?, break, by, class, companion, continue, data, do, dynamic, else, enum, expect, false, final, for, fun, get, if, import, in, interface, internal, is, null, object, override, package, private, public, return, set, super, suspend, this, throw, true, try, typealias, val, var, vararg, when, where, while,
  inner, open},
  morecomment=[s]{/*}{*/},
  morestring=[b]",
  morestring=[s]{"""*}{*"""},
  ndkeywords={@Deprecated, @JvmField, @JvmName, @JvmOverloads, @JvmStatic, @JvmSynthetic, Array, Byte, Double, Float, Int, Integer, Iterable, Long, Runnable, Short, String},
  ndkeywordstyle={\color[rgb]{0,0,0}\bfseries},
  sensitive=true,
  stringstyle={\color[rgb]{0,0,1}\ttfamily},
}

\newcommand{\tool}{{\sc thalia}}
\newcommand{\heph}{{\sc hephaestus}}
\newcommand{\ir}{API-IR}
\newcommand{\graph}{API graph}
\newcommand{\empirical}[1]{#1}

\newcommand{\point}[1]{\par\smallskip\noindent{\textbf{#1:}} }

\newcommand{\libaname}{com.fasterxml.jackson.core:jackson-core\xspace}
\newcommand{\libanodes}{\nnum{6775}\xspace}
\newcommand{\libaedges}{\nnum{10158}\xspace}
\newcommand{\libapolym}{\nnum{16}\xspace}
\newcommand{\libamethods}{\nnum{1442}\xspace}
\newcommand{\libafields}{\nnum{54}\xspace}
\newcommand{\libacons}{\nnum{82}\xspace}
\newcommand{\libatypecon}{\nnum{4}\xspace}
\newcommand{\libatypes}{\nnum{98}\xspace}
\newcommand{\libainh}{\nnum{3.80}\xspace}
\newcommand{\libasig}{\nnum{2.47}\xspace}
\newcommand{\libbname}{com.google.guava:guava\xspace}
\newcommand{\libbnodes}{\nnum{10235}\xspace}
\newcommand{\libbedges}{\nnum{14469}\xspace}
\newcommand{\libbpolym}{\nnum{975}\xspace}
\newcommand{\libbmethods}{\nnum{4363}\xspace}
\newcommand{\libbfields}{\nnum{188}\xspace}
\newcommand{\libbcons}{\nnum{0}\xspace}
\newcommand{\libbtypecon}{\nnum{214}\xspace}
\newcommand{\libbtypes}{\nnum{412}\xspace}
\newcommand{\libbinh}{\nnum{3.31}\xspace}
\newcommand{\libbsig}{\nnum{2.68}\xspace}
\newcommand{\libcname}{org.apache.commons:commons-lang3\xspace}
\newcommand{\libcnodes}{\nnum{8559}\xspace}
\newcommand{\libcedges}{\nnum{11514}\xspace}
\newcommand{\libcpolym}{\nnum{272}\xspace}
\newcommand{\libcmethods}{\nnum{2643}\xspace}
\newcommand{\libcfields}{\nnum{307}\xspace}
\newcommand{\libccons}{\nnum{193}\xspace}
\newcommand{\libctypecon}{\nnum{93}\xspace}
\newcommand{\libctypes}{\nnum{211}\xspace}
\newcommand{\libcinh}{\nnum{2.80}\xspace}
\newcommand{\libcsig}{\nnum{2.60}\xspace}
\newcommand{\libdname}{org.apache.logging.log4j:log4j-api\xspace}
\newcommand{\libdnodes}{\nnum{7207}\xspace}
\newcommand{\libdedges}{\nnum{10805}\xspace}
\newcommand{\libdpolym}{\nnum{49}\xspace}
\newcommand{\libdmethods}{\nnum{1771}\xspace}
\newcommand{\libdfields}{\nnum{66}\xspace}
\newcommand{\libdcons}{\nnum{122}\xspace}
\newcommand{\libdtypecon}{\nnum{12}\xspace}
\newcommand{\libdtypes}{\nnum{134}\xspace}
\newcommand{\libdinh}{\nnum{3.77}\xspace}
\newcommand{\libdsig}{\nnum{4.05}\xspace}
\newcommand{\libename}{org-assertj-assertj-core\xspace}
\newcommand{\libenodes}{\nnum{11321}\xspace}
\newcommand{\libeedges}{\nnum{16169}\xspace}
\newcommand{\libepolym}{\nnum{633}\xspace}
\newcommand{\libemethods}{\nnum{5250}\xspace}
\newcommand{\libefields}{\nnum{179}\xspace}
\newcommand{\libecons}{\nnum{0}\xspace}
\newcommand{\libetypecon}{\nnum{151}\xspace}
\newcommand{\libetypes}{\nnum{656}\xspace}
\newcommand{\libeinh}{\nnum{4.41}\xspace}
\newcommand{\libesig}{\nnum{2.67}\xspace}
\newcommand{\libfname}{org.clojure:clojure\xspace}
\newcommand{\libfnodes}{\nnum{9032}\xspace}
\newcommand{\libfedges}{\nnum{12532}\xspace}
\newcommand{\libfpolym}{\nnum{1}\xspace}
\newcommand{\libfmethods}{\nnum{2676}\xspace}
\newcommand{\libffields}{\nnum{466}\xspace}
\newcommand{\libfcons}{\nnum{169}\xspace}
\newcommand{\libftypecon}{\nnum{2}\xspace}
\newcommand{\libftypes}{\nnum{624}\xspace}
\newcommand{\libfinh}{\nnum{3.85}\xspace}
\newcommand{\libfsig}{\nnum{3.23}\xspace}
\newcommand{\libgname}{org.mockito:mockito-core\xspace}
\newcommand{\libgnodes}{\nnum{5783}\xspace}
\newcommand{\libgedges}{\nnum{8048}\xspace}
\newcommand{\libgpolym}{\nnum{98}\xspace}
\newcommand{\libgmethods}{\nnum{510}\xspace}
\newcommand{\libgfields}{\nnum{9}\xspace}
\newcommand{\libgcons}{\nnum{0}\xspace}
\newcommand{\libgtypecon}{\nnum{27}\xspace}
\newcommand{\libgtypes}{\nnum{136}\xspace}
\newcommand{\libginh}{\nnum{3.39}\xspace}
\newcommand{\libgsig}{\nnum{2.28}\xspace}
\newcommand{\libhname}{groovy-stdlib\xspace}
\newcommand{\libhnodes}{\nnum{14817}\xspace}
\newcommand{\libhedges}{\nnum{22303}\xspace}
\newcommand{\libhpolym}{\nnum{461}\xspace}
\newcommand{\libhmethods}{\nnum{10982}\xspace}
\newcommand{\libhfields}{\nnum{1108}\xspace}
\newcommand{\libhcons}{\nnum{1159}\xspace}
\newcommand{\libhtypecon}{\nnum{164}\xspace}
\newcommand{\libhtypes}{\nnum{1276}\xspace}
\newcommand{\libhinh}{\nnum{4.21}\xspace}
\newcommand{\libhsig}{\nnum{2.46}\xspace}
\newcommand{\libiname}{kotlin-stdlib\xspace}
\newcommand{\libinodes}{\nnum{6020}\xspace}
\newcommand{\libiedges}{\nnum{8986}\xspace}
\newcommand{\libipolym}{\nnum{205}\xspace}
\newcommand{\libimethods}{\nnum{3808}\xspace}
\newcommand{\libifields}{\nnum{1108}\xspace}
\newcommand{\libicons}{\nnum{438}\xspace}
\newcommand{\libitypecon}{\nnum{70}\xspace}
\newcommand{\libitypes}{\nnum{395}\xspace}
\newcommand{\libiinh}{\nnum{5.47}\xspace}
\newcommand{\libisig}{\nnum{2.37}\xspace}
\newcommand{\libjname}{scala-stdlib\xspace}
\newcommand{\libjnodes}{\nnum{3933}\xspace}
\newcommand{\libjedges}{\nnum{5928}\xspace}
\newcommand{\libjpolym}{\nnum{390}\xspace}
\newcommand{\libjmethods}{\nnum{2439}\xspace}
\newcommand{\libjfields}{\nnum{688}\xspace}
\newcommand{\libjcons}{\nnum{246}\xspace}
\newcommand{\libjtypecon}{\nnum{93}\xspace}
\newcommand{\libjtypes}{\nnum{331}\xspace}
\newcommand{\libjinh}{\nnum{4.23}\xspace}
\newcommand{\libjsig}{\nnum{2.49}\xspace}
\newcommand{\libkname}{other\xspace}
\newcommand{\libknodes}{\nnum{6900}\xspace}
\newcommand{\libkedges}{\nnum{9974}\xspace}
\newcommand{\libkpolym}{\nnum{43}\xspace}
\newcommand{\libkmethods}{\nnum{1330}\xspace}
\newcommand{\libkfields}{\nnum{150}\xspace}
\newcommand{\libkcons}{\nnum{97}\xspace}
\newcommand{\libktypecon}{\nnum{16}\xspace}
\newcommand{\libktypes}{\nnum{204}\xspace}
\newcommand{\libkinh}{\nnum{3.32}\xspace}
\newcommand{\libksig}{\nnum{2.37}\xspace}
\newcommand{\avgname}{Avg\xspace}
\newcommand{\avgnodes}{\nnum{7052}\xspace}
\newcommand{\avgedges}{\nnum{10192}\xspace}
\newcommand{\avgpolym}{\nnum{71}\xspace}
\newcommand{\avgmethods}{\nnum{1563}\xspace}
\newcommand{\avgfields}{\nnum{177}\xspace}
\newcommand{\avgcons}{\nnum{112}\xspace}
\newcommand{\avgtypecon}{\nnum{23}\xspace}
\newcommand{\avgtypes}{\nnum{227}\xspace}
\newcommand{\avginh}{\nnum{3.38}\xspace}
\newcommand{\avgsig}{\nnum{2.4}\xspace}

\newcommand{\gconfirmed}{\nnum{38}\xspace}
\newcommand{\gfix}{\nnum{20}\xspace}
\newcommand{\gwont}{\nnum{2}\xspace}
\newcommand{\gdupl}{\nnum{2}\xspace}

\newcommand{\gcrash}{\nnum{8}\xspace}
\newcommand{\gucte}{\nnum{47}\xspace}
\newcommand{\gurb}{\nnum{6}\xspace}
\newcommand{\gcpi}{\nnum{1}\xspace}

\newcommand{\ggenerator}{\nnum{32}\xspace}
\newcommand{\gsoundness}{\nnum{9}\xspace}
\newcommand{\ginference}{\nnum{21}\xspace}

\newcommand{\gtotal}{\nnum{62}\xspace}

\newcommand{\kconfirmed}{\nnum{9}\xspace}
\newcommand{\kfix}{\nnum{1}\xspace}
\newcommand{\kwont}{\nnum{1}\xspace}
\newcommand{\kdupl}{\nnum{0}\xspace}

\newcommand{\kcrash}{\nnum{3}\xspace}
\newcommand{\kucte}{\nnum{7}\xspace}
\newcommand{\kurb}{\nnum{0}\xspace}
\newcommand{\kcpi}{\nnum{1}\xspace}

\newcommand{\kgenerator}{\nnum{8}\xspace}
\newcommand{\ksoundness}{\nnum{0}\xspace}
\newcommand{\kinference}{\nnum{3}\xspace}

\newcommand{\ktotal}{\nnum{11}\xspace}

\newcommand{\sconfirmed}{\nnum{8}\xspace}
\newcommand{\sfix}{\nnum{1}\xspace}
\newcommand{\swont}{\nnum{1}\xspace}
\newcommand{\sdupl}{\nnum{1}\xspace}

\newcommand{\scrash}{\nnum{4}\xspace}
\newcommand{\sucte}{\nnum{7}\xspace}
\newcommand{\surb}{\nnum{0}\xspace}
\newcommand{\scpi}{\nnum{0}\xspace}

\newcommand{\sgenerator}{\nnum{7}\xspace}
\newcommand{\ssoundness}{\nnum{1}\xspace}
\newcommand{\sinference}{\nnum{3}\xspace}

\newcommand{\stotal}{\nnum{11}\xspace}

\newcommand{\tconfirmed}{\nnum{55}\xspace}
\newcommand{\tfix}{\nnum{22}\xspace}
\newcommand{\twont}{\nnum{4}\xspace}
\newcommand{\tdupl}{\nnum{3}\xspace}

\newcommand{\tcrash}{\nnum{15}\xspace}
\newcommand{\tucte}{\nnum{61}\xspace}
\newcommand{\turb}{\nnum{6}\xspace}
\newcommand{\tcpi}{\nnum{2}\xspace}

\newcommand{\tgenerator}{\nnum{47}\xspace}
\newcommand{\tsoundness}{\nnum{10}\xspace}
\newcommand{\tinference}{\nnum{27}\xspace}

\newcommand{\treal}{\nnum{77}\xspace}
\newcommand{\ttotal}{\nnum{84}\xspace}

\newcommand{\parameterizedclass}{\nnum{58}\xspace}
\newcommand{\parameterizedtype}{\nnum{58}\xspace}
\newcommand{\overloading}{\nnum{22}\xspace}
\newcommand{\llambda}{\nnum{10}\xspace}
\newcommand{\parameterizedfunction}{\nnum{45}\xspace}
\newcommand{\sam}{\nnum{30}\xspace}
\newcommand{\staticmethod}{\nnum{2}\xspace}
\newcommand{\wildcardtype}{\nnum{19}\xspace}
\newcommand{\functionreference}{\nnum{21}\xspace}
\newcommand{\boundedtypeparameter}{\nnum{15}\xspace}
\newcommand{\arraytype}{\nnum{7}\xspace}
\newcommand{\inheritance}{\nnum{8}\xspace}
\newcommand{\subtyping}{\nnum{7}\xspace}
\newcommand{\variabletypeinference}{\nnum{2}\xspace}
\newcommand{\typeargumentinference}{\nnum{23}\xspace}
\newcommand{\recursiveupperbound}{\nnum{3}\xspace}
\newcommand{\variableargument}{\nnum{5}\xspace}
\newcommand{\primitivetype}{\nnum{3}\xspace}
\newcommand{\bridgemethod}{\nnum{1}\xspace}
\newcommand{\defaultmethod}{\nnum{1}\xspace}
\newcommand{\conditionals}{\nnum{4}\xspace}
\newcommand{\innerclass}{\nnum{3}\xspace}
\newcommand{\accessmodifier}{\nnum{2}\xspace}
\newcommand{\nullabletype}{\nnum{2}\xspace}
\newcommand{\operator}{\nnum{1}\xspace}

\newcommand{\groovylocavg}{\nnum{13}\xspace}
\newcommand{\groovysizeavg}{\nnum{1.6}\xspace}
\newcommand{\scalalocavg}{\nnum{11}\xspace}
\newcommand{\scalasizeavg}{\nnum{1.6}\xspace}
\newcommand{\kotlinlocavg}{\nnum{11}\xspace}
\newcommand{\kotlinsizeavg}{\nnum{1.4}\xspace}

\newcommand{\groovylinemissed}{\nnum{1118}\xspace}
\newcommand{\groovybranchmissed}{\nnum{5395}\xspace}
\newcommand{\groovyfunctionmissed}{\nnum{173}\xspace}

\newcommand{\scalabranchmissed}{\nnum{22334}\xspace}
\newcommand{\scalafunctionmissed}{\nnum{654}\xspace}

\newcommand{\kotlinlinemissed}{\nnum{2567}\xspace}
\newcommand{\kotlinbranchmissed}{\nnum{17548}\xspace}
\newcommand{\kotlinfunctionmissed}{\nnum{581}\xspace}

\begin{document}
\sloppy

\title{Extended Paper: API-driven Program Synthesis for Testing Static Typing Implementations}

\author{Thodoris Sotiropoulos}
\affiliation{%
  \institution{ETH Zurich}
  \country{Switzerland}
}
\email{theodoros.sotiropoulos@inf.ethz.ch}

\author{Stefanos Chaliasos}
\affiliation{%
  \institution{Imperial College London}
  \country{UK}
}
\email{s.chaliasos21@imperial.ac.uk}

\author{Zhendong Su}
\affiliation{%
  \institution{ETH Zurich}
  \country{Switzerland}
}
\email{zhendong.su@inf.ethz.ch}

\begin{abstract}

We introduce a novel approach
for testing static typing implementations
based on the concept of
{\it API-driven program synthesis}.
The idea is to synthesize type-intensive but
small and well-typed programs
by leveraging and combining
{\it application programming interfaces (APIs)}
derived from existing software libraries.
Our primary insight is backed up by real-world evidence:
a significant number of compiler typing bugs are caused by
small test cases that employ APIs
from the standard library of the language under test.
This is attributed to the inherent complexity
of the majority of these APIs,
which often exercise a wide range of sophisticated
type-related features.
The main contribution of our approach
is the ability to produce small client programs
with increased feature coverage,
without bearing the burden of generating 
the corresponding well-formed API definitions from scratch.
To validate diverse aspects of static typing procedures
(i.e., soundness, precision of type inference),
we also enrich our API-driven approach
with fault-injection and semantics-preserving
modes, along with their corresponding test oracles.

We evaluate our implemented tool, \tool,
on testing the static typing implementations
of the compilers for three popular languages,
namely, Scala, Kotlin, and Groovy.
\tool~has uncovered~$\ttotal$ typing bugs
($\treal$ confirmed and $\tfix$ fixed),
most of which are triggered by
test cases featuring APIs that rely
on parametric polymorphism,
overloading,
and higher-order functions.
Our comparison with state-of-the-art shows 
that~\tool~yields test programs with
distinct characteristics,
offering additional and complementary benefits.

\end{abstract}



\maketitle

\section{Introduction}
\label{sec:intro}

Type safety is a fundamental property
contributing to the correct
execution of computer programs.
{\it Static typing},
an integral process
in every statically-typed programming language
implementation (typically part of the compiler),
verifies that a source program is type-correct
and type-safe based on a {\it type system}.
A type system lies at the heart
of the design of a language---it outlines
a set of rules regarding the language's types,
and how these types can be used and combined~\cite{types-and-langs}.
Language designers and researchers strive
to build sound type systems~\cite{taming-wildcard,type-protocol,sound-concurrency},
which guarantee type safety,
i.e., being able to identify
all potential type errors
during compilation.

At the same time,
modern languages continuously evolve
with new (and often sophisticated)
features to provide users
with a smoother programming experience.
However,
integrating these new features
into a language poses
significant challenges and considerations
for implementing sound and practical type systems.
Indeed,
recent work~\cite{typing-study,hephaestus} has shown
that the static typing implementations of
well-established compilers
suffer from a substantial number of bugs~\cite{java-scala-unsoundness}.
These implementation flaws lead to
(1) reliability and security ramifications
on the programs compiled with the faulty compilers,
and (2) degradation of the programmers’ experience and productivity.
In particular,
compiler typing bugs
typically cause
frustrating rejections of well-typed programs,
dangerous acceptances of erroneous, type-unsafe programs,
or annoying crashes
and compilation performance issues.

In spite of the sharp rise of compiler
testing research,
improving the reliability of compilers' type checkers
has been until recently a neglected problem.
Indeed,
most of the existing
compiler testing techniques
target optimizing compilers~\cite{csmith,emi,skeletal,shader,yarpgen,yarpgen2}.
As shown in the empirical work of~\citet{typing-study},
detecting typing bugs
demands new compiler testing methods,
because typing bugs exhibit
distinct characteristics
(e.g., symptoms, root causes)
compared to optimization bugs.

Currently,
there are only a few approaches 
to validating static typing implementations.
Focusing on Rust's type checker,
\citet{rust} have introduced a fuzzing approach
based
on constraint logic programming.
More recently,
\citet{hephaestus} have developed~\heph,
which produces type-intensive programs written
in an intermediate language (IR) that supports
parametric polymorphism and type inference.
This high-level IR allows
targeting multiple languages
(i.e., Java, Groovy, and Kotlin).

All the aforementioned approaches rely on
\emph{generative compiler testing}:
constructing programs entirely from scratch
based on the syntactic and semantic rules
of a language under test.
However,
generative compiler testing involves a significant limitation:
it is unable to test features beyond
what the underlying code generator can handle
and synthesize.
Therefore,
the utility of generative techniques
can easily saturate~\cite{saturation},
producing programs that exhibit the same programming idioms,
unless new language constructs are implemented
in the program generator.

\point{Approach}
We introduce a novel approach
for testing static typing implementations 
based on the concept of \emph{API-driven program synthesis}.
Our approach is motivated by
the findings of a recent
empirical study~\cite{typing-study}:
around one third
of compiler typing bugs
are triggered by
small test cases that employ
{\it application programming interfaces (APIs)}
from the standard library of the language being tested.
This finding is well-justified,
considering the inherent complexity
of the included API definitions
(e.g., functions, variables, types),
which heavily rely on
advanced type-related features
such as parametric polymorphism.
Building upon this finding,
our approach leverages the abundance of existing APIs,
and produces small (yet complex) client programs,
\emph{without} the burden of producing
the corresponding well-formed definitions from scratch.
Our approach exploits the fact that
although libraries are pre-compiled,
a compiler still performs type checking on
every client program that refers
to components of a library API.
This ensures that the features provided by the API
are used in a type-safe manner.

At the high level,
our approach works as follows.
The input is an API,
which is modeled as an~\emph{\graph},
a structure that captures
the dependencies and relations among
API components.
Then,
our approach proceeds with the concept of~\emph{API enumeration}:
for every component $d$ (i.e., method/field) of the input API,
it systematically explores
all unique, type-safe usages of $d$
w.r.t. the signature of $d$.
Invoking an API component
through different typing patterns
lets us exercise the implementation
of many type-related operations in the compiler
(e.g., subtyping rules).
The outcome of the API enumeration process
is a finite set of {\it abstract typed expressions}.
An abstract typed expression
is encoded as a sequence of types
that represents which specific types are combined
to employ a certain API entity.
In turn,
our approach concretizes
each abstract typed expression
by replacing each type in the sequence
with an {\it inhabitant}~\cite{type-inhabitant}
that re-uses code from the input library.
To synthesize an inhabitant,
the approach consults the~\graph~to
identify sequences of method calls or field accesses
that match the given type.
By default,
our method
yields well-typed client programs.
Consequently,
a compiler is expected to
successfully
accept these generated programs.
Failures to do so
indicate potential bugs
in the compiler.

Our method is also equipped with two modes
that allow the discovery of type inference
and soundness bugs respectively.
Specifically,
the first mode produces well-typed programs
with omitted type information.
To do so,
when encountering a polymorphic call,
our method provides no explicit type arguments
whenever these can be inferred by
the surrounding context.
The second mode
produces ill-typed programs
by enumerating all those abstract-typed expressions
that employ a specific API component in a way
that violates typing rules.
This mode enables the detection of soundness bugs
by finding cases where
the compiler mistakenly
accepts ill-typed programs.

\point{Results}
Our implementation,
which we call~\tool,\footnote{In Greek mythology,
Thalia was a nymph daughter of the smithing god Hephaestus.}
produces programs written
in three popular languages:
Scala, Groovy, and Kotlin.
In our evaluation,
we collected a large corpus of popular APIs taken
from Maven's central software repository.
Based on the collected APIs,
\tool~produced small client programs
that were able to trigger~$\ttotal$ unique bugs,
of which~$\treal$ have been either confirmed
or fixed.
Our two modes
helped identify~$\tinference$
type inference bugs,
and~$\tsoundness$ bugs triggered by wrongly-typed code.
When comparing our work with the state-of-the-art
framework~\heph~\cite{hephaestus},
we find that,
despite producing test programs
one order of magnitude smaller than prior work,
\tool~is able to detect at least~\empirical{42}~bugs missed by it.
Furthermore,
\tool's programs exercise previously unexplored compiler regions,
leading to an increase in \heph' code coverage
by up to $\empirical{9\%}$ (\empirical{\nnum{\kotlinlinemissed}}) for lines of code,
\empirical{10\%} (\empirical{\nnum{\kotlinbranchmissed}}) for branches,
and~\empirical{8\%} (\empirical{\nnum{\kotlinfunctionmissed}}) for functions.

\point{Contributions}
Our work makes the following contributions.
\begin{itemize}
\item The notion of API enumeration,
    which systematically examines API component usages
    via diverse valid (or invalid) typing patterns
    (aka abstract-typed expressions).
    This allows the exploration of complex type-related functionalities
    in the compiler,
    making our approach suitable for testing.
\item A novel API-driven program synthesis
    approach for producing small yet complex
    client programs based on well-typed (or ill-typed)
    abstract-typed expressions via API enumeration.
\item An open-source implementation called~\tool,
    which is able to produce client code
    written in three different languages:
    Scala, Kotlin, and Groovy.
\item An extensive evaluation of~\tool~covering several aspects,
    including
    bug-finding capability,
    characteristics of test cases,
    code coverage,
    and comparison with state-of-the-art.
    Overall,
    \tool~uncovered~$\ttotal$ new faults triggered
    in three widely-used compilers.
    These failures were caused by a plethora of language features
    available in the input APIs.
\end{itemize}

\point{Availability}
\tool\ is available as open-source software
under the GNU General Public License v3.0
at~\url{https://github.com/hephaestus-compiler-project/thalia}.
The research artifact is available
at~\url{https://zenodo.org/record/8425071}.

\section{Background and Illustrative Examples}
\label{sec:background}

This section
presents the background
and motivation of our API-driven approach
by discussing the limitations of
the existing techniques in generative compiler testing,
especially those of~\heph,
and how we address them.

\point{Limitations of Hephaestus}
\heph~\cite{hephaestus} is the state-of-the-art tool
for finding compiler typing bugs.
It adopts a generative process similar to that of Csmith~\cite{csmith}.
Specifically,
\heph\ creates a bunch of random class/method/field definitions
that respect the syntax and the semantics of a target language
(in this case~\heph' IR).
This generative process involves three major limitations:
\begin{itemize}
\item The generated programs contain~\emph{only}
a limited set of features. 
This means that the program generator is unable
to exercise a certain feature, if this feature is not
supported by the implementation.
Extending the program generator with new features
is technically hard,
and often only works for a specific
language.
This is because the program generator needs to apply
a set of language-specific semantic checks
to ensure the validity of
the generated programs.
For example,
creating a well-formed class that
implements multiple interfaces
requires checking that the interfaces
have no conflicting method signatures.

\item Even after using diverse configurations,
the generative process easily comes to
a saturation point~\cite{saturation},
because the resulting programs are somewhat biased in
generating programs that exhibit the same programming idioms.

\item To test a new language (e.g., Rust),
someone typically needs to engineer a new program generator.
\end{itemize}

\begin{figure}[t]
\centering
\begin{subfigure}{1\linewidth}
\begin{lstlisting}[language=scala]
import java.util.LinkedList
import com.google.common.collect.LinkedHashMultiset
def test(): Unit {
    val x: Tuple1[LinkedHashMultiset[String]] = ???
    val res: Any = x._1.retainAll(new LinkedList())
}
\end{lstlisting}
\vspace{-3mm}
\caption{Test case written in Scala.}
\label{fig:scala-test-case}
\end{subfigure}
\hfill
\begin{subfigure}{1\linewidth}
\begin{lstlisting}[language=java]
package com.google.common.collect;
public class LinkedHashMultiset<T> extends AbstractMultiSet<T> {  }
class AbstractMultiSet<T> implements MultiSet<T> {
  public void retailAll(Collection<?> p) {}
}
\end{lstlisting}
\vspace{-2mm}
\caption{Definition of the LinkedHashMultiset API in the~\href{https://github.com/google/guava}{guava} library.}
\label{fig:scala-ex-def}
\end{subfigure}
\vspace{-4mm}
\caption{\href{https://github.com/lampepfl/dotty/issues/17391}{DOTTY-17391}:
This program triggers a crash in Dotty.
The program exercises the~\href{https://github.com/google/guava}{guava} API.
}
\label{fig:scala-ex}
\vspace{-4mm}
\end{figure}

\point{The benefits and power of APIs}
To tackle these limitations of~\heph,
we introduce a complementary approach
that relies on real-world~\emph{APIs},
which are well-engineered,
expressive,
and utilize diverse language features.
Rather than generating large programs
consisting of random definitions,
we synthesize small but intricate client programs
that invoke components (e.g., method)
from a given API.
This allows us to
(1) effortlessly exercise a rich set of features
without the added complexity of creating complex
definitions from scratch,
and (2) combine individual features in interesting and unexpected ways.
In contrast to
conventional generative processes
(see~\heph),
API-based test programs
provide the following unique benefits:
\begin{itemize}

\item {\it Rich feature coverage for free}:
As APIs are typically designed for use in a wide variety of contexts,
they often combine diverse sophisticated and demanding language features,
such as bounded polymorphism,
subtyping,
or overloading.
Given that their interfaces are also likely to be well-tested and widely used,
they are free from type errors.
Therefore, a feature-agnostic and potent generation process can be achieved
by creating test programs centered around APIs
without creating these complex and well-typed API definitions ourselves.
Furthermore,
the huge variety of real-world APIs
avoids,
or at least delays,
the saturation of the generation process.

\item {\it Applicability}:
APIs are ubiquitous in mainstream languages.
Generating type-intensive programs for
a new language mostly requires the collection of its APIs
(see Section~\ref{general}---Generalizability).

\item {\it Efficient testing}:
Generating small programs improves the throughput of testing,
as the compilers under test require considerably
less time to process small inputs than large ones.

\item {\it Simple test-case reduction}:
Test-case reduction is highly important
for compiler testing campaigns~\cite{c-reduce}.
By construction,
API-driven program synthesis
yields small, self-contained programs
that require minimal test-case reduction.

\end{itemize}

Consider the following two bugs found by
API-driven programs synthesized by our tool.

\point{Internal compiler error in the Scala 3 compiler}
Figure~\ref{fig:scala-test-case}
shows a valid program
that triggers a crash in the compiler of Scala 3,
also known as Dotty.
The program imports
the {\tt LinkedHashMultiset} class
from \href{https://github.com/google/guava}{\tt com.google.guava:guava}
(line 2),
a popular library
that extends the Java collections framework.
On line 4,
the code uses the Scala standard library
to create
a tuple containing an element
of type {\tt LinkedHashMultiset<String>}.
Next,
by accessing this element,
the program invokes the method {\tt retainAll}
found in the {\tt LinkedHashMultiset} class.

Figure~\ref{fig:scala-ex-def}
shows the {\tt LinkedHashMultiset} class
as defined in the {\tt guava} library.
The class inherits {\tt retainAll} from base class
{\tt AbstractMultiSet}.
Notably,
{\tt AbstractMultiSet} involves a {\tt default}
access modifier (line 3),
meaning that the class is inaccessible
from code outside the package
{\tt com.google.common.collect}.
However,
its method {\tt retainAll} features
a {\tt public} access modifier (line 4),
which suggests that the method can be accessed
by any public subclass of {\tt AbstractMultiSet},
including {\tt LinkedHashMultiset}.
Although {\tt retainAll}
belongs to the API of {\tt LinkedHashMultiset},
Dotty mistakenly treats the method as inaccessible,
which in turn leads to a compiler crash
while typing the method call on line 5 (Figure~\ref{fig:scala-test-case}).
Surprisingly,
replacing the receiver expression
{\tt x.\_1} with any other expression
of type {\tt LinkedHashMultiset<String>}
successfully compiles the program.

\begin{figure}
\centering
\begin{tikzpicture}[thick, scale=0.82,
every node/.style={scale=0.8},
font=\small\sffamily, node distance={0.4cm},
main/.style = {draw=none, rectangle}] 
\node[main](0) at (1,0) {
\begin{lstlisting}[language=Java,linewidth=5cm, xleftmargin=0cm]
import org.apache.commons.lang3.ArrayUtils;
class Test {
  void test() {
    byte x = new byte[0];
    byte[] res = ArrayUtils.removeAll(x, 0);
  }
}
\end{lstlisting}
};
\node[main] (1) at (7, 0.3) {
    \begin{lstlisting}[language=Java, numbers=none]
    package org.apache.commons.lang3.ArrayUtils;
    class ArrayUtils {
      static boolean[] removeAll(boolean[] array, int... indices)
      static byte[] removeAll(byte[] array, int... indices)
      static short[] removeAll(short[] array, int... indices)
      static int[] removeAll(int[] array, int... indices)
      static long[] removeAll(long[] array, int... indices)
      static float[] removeAll(float[] array, int... indices)
      static double[] removeAll(double[] array, int... indices)
      static char[] removeAll(char[] array, int... indices)
      static <T> T[] removeAll(T[] array, int... indices)
    }
    \end{lstlisting}

}; 
\draw[->] (1.6,-0.25) -- (2.8,1.1);
\draw (3,1) rectangle (11.1,1.4);
\end{tikzpicture}
\vspace{-3mm}
\caption{\href{https://issues.apache.org/jira/browse/GROOVY-11053}{GROOVY-11053}:
A well-typed program rejected by the Groovy compiler.
The program exercises the \href{https://github.com/apache/commons-lang}{apache-commons-lang3} library.}
\label{fig:groovy-test-case}
\vspace{-3mm}
\end{figure}

\point{Unexpected method ambiguity error in {\tt groovyc}}
Figure~\ref{fig:groovy-test-case}
shows another bug
where {\tt groovyc} erroneously rejects
a well-typed program.
The code first defines an array of {\tt byte} (line 4),
and then calls a static method
from the {\tt ArrayUtils} API of
the~\href{https://github.com/apache/commons-lang}{\tt org.apache.commons:commons-lang3} library
to remove the element of the array at position 0 (line 5).

The API of {\tt ArrayUtils} is quite complex:
it provides nine overloaded methods all named {\tt removeAll}.
Eight of them operate on primitive arrays,
while the remaining one is a polymorphic method
operating on arrays of reference types.
When a compiler encounters a call
to an overloaded method,
it chooses to invoke the most specific one
based on a set of criteria,
such as the types of the provided arguments.
Although it is clear from the context
that the code intends to call the
method variant that takes an array of bytes,
a bug in {\tt groovyc}'s method resolution
causes the compiler to accidentally reject the program
with an error of the form:
\textit{``reference to method removeAll is ambiguous.
Cannot choose between candidate methods.''}
The root cause of the failure
lies in the way variable arguments
are matched against formal parameter types
related to primitive types.

As the two examples illustrate,
although the invoked APIs seem simple
at a first glance,
it is in fact challenging for the compilers
to correctly handle these APIs.
Indeed,
the {\tt guava} API exercises access modifiers
and complex inheritance scenarios,
while the {\tt apache-commons-lang3} API
has many overloaded methods
mixed with plenty of type-related features
(e.g., primitive types,
array types,
and variable arguments).
Interestingly,
API components might hide other typing features.
For example,
although calling method {\tt retainAll} in Figure~\ref{fig:scala-test-case}
knows nothing about what an access modifier stands for,
or the inheritance chain of the receiver,
the corresponding test case triggers
a compiler bug associated with these ``hidden'' features.

\section{API-driven Program Synthesis}
\label{sec:approach}

\begin{figure}[t]
\includegraphics[scale=0.44]{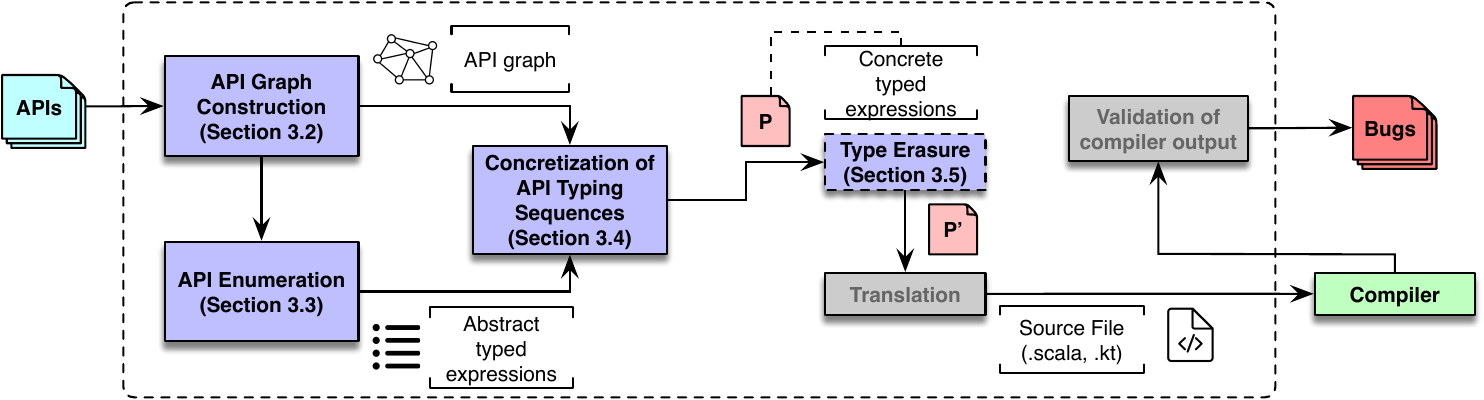}
\vspace{-3mm}
\caption{The API-driven program synthesis approach
for finding compiler typing bugs.}
\label{fig:arch}
\vspace{-4mm}
\end{figure}

Figure~\ref{fig:arch} gives an overview
of our API-driven program synthesis approach.
The process begins with a corpus of APIs
extracted from either the standard library of the language being tested
or its third-party libraries.
These APIs are represented in an~\emph{\graph}
which captures the dependencies and type signature
of each API component (Section~\ref{sec:api-graph}).
Next,
our approach performs \emph{API enumeration}
which systematically explores all possible invocations
of the encompassed API components
through unique typing patterns (Section~\ref{sec:api-enumeration}).
A typing pattern is
a sequence of types
that corresponds to~\emph{abstract typed expressions}.
Conceptually,
an abstract typed expression represents
a combination of types
used to invoke a particular API entity (e.g, method).
An abstract expression can be
either well-typed or ill-typed
with regards to the type signature of
the corresponding API component.
Then,
the approach yields well-typed or ill-typed programs
by concretizing each abstract-typed expression
into a concrete one written
in a reference language called~\ir.
To do so,
it examines the~\graph~to find type inhabitants 
by enumerating the set of paths
that reach a specific type node
via standard graph reachability algorithms
(Section~\ref{sec:concretization}).

As an optional step,
our approach employs
the type erasure process
whose purpose is to remove the type arguments
of polymorphic calls,
while maintaining the type correctness of the
original expression (Section~\ref{sec:type-erasure}).
The insight is that
constructing polymorphic calls without explicit type arguments
helps
exercise the implementations of type inference operations.
The final step is to employ
language-specific translators that
transform an expression in~\ir~
into a source file written in an actual
language (e.g., Scala).
Ultimately,
the generated source files are given
as input to the compiler under test,
whose output is checked against
the given oracle for potential bugs.
In particular,
we expect the compiler to accept
the programs derived from well-typed typing sequences,
and reject those
that come from ill-typed ones.
We now describe
our approach in detail.

\subsection{Preliminary Definitions}
\label{sec:preliminary-defs}

Figure~\ref{fig:ir} presents the core language (called~\ir)
that we use as a base to explain our approach.
\ir~is a basic language
that supports three major programming constructs:
classes, functions, and fields/variables.
\footnote{In
our implementation, the core language is extended with
more sophisticated features,
such as lambdas, or wildcard types.}
\ir~is also equipped with parametric polymorphism
in the style of Java-like generics.
The~\ir~is designed with simplicity in mind,
ensuring that the principles of our approach
(discussed in subsequent sections)
can be adapted to any language compatible with~\ir.
In what follows,
we use the vector notation
to represent a sequence of elements.
For example,
$\overline{t}$ means a sequence of $t$ elements.

\ir~is parameterized by an API denoted as $\Lambda \in API$.
An API is a set
consisting of (polymorphic) classes,
(polymorphic) functions,
and field/variables.
A class defined as $\prim{class}\ \mathcal{C}\prim{<}\overline{\alpha}\prim{>}\ \prim{extends}\ \overline{t}: \overline{v}$
defines a sequence of formal type variables $\overline{\alpha}$,
and extends a sequence of types $\overline{t}$.
The body of a class represented by $\overline{v}$ contains
other API members,
including nested class definitions.
A function $\prim{fun}\ m\prim{<}\overline{\alpha}\prim{>}(\overline{x: t_1}): t_2$
with name $m$:
(1) introduces a set of type variables $\overline{\alpha}$,
(2) takes a sequence of formal parameters,
and (3) outputs a value of type $t_2$.
A class or a function
is considered~\emph{polymorphic} only when
the sequence $\overline{\alpha}$ in their syntax is not empty.
A program in~\ir~is a sequence of expressions
that use components from the given API $\Lambda$.
An expression is either a constant of type $t \in \textit{Type}$
denoted as $\textit{constant}(t)$,
a local variable definition,
a field access,
or a function call.
The expression $\epsilon$ represents
an empty expression
used to model the invocation of top-level methods
and fields,
e.g., methods with no explicit receiver.

The type system in~\ir~is nominal.
The set of types
includes the usual $\bot$ and $\top$ types,
or a type variable $\phi: t$ with an upper bound $t$.
A class type $\mathcal{T}\prim{<}\overline{\alpha}\prim{>}: \overline{t}$
is labeled with a name $\mathcal{T}$,
a sequence of type variables $\overline{\alpha}$,
and a sequence of supertypes $\overline{t}$.
A class type is derived from a class $c \in \textit{Class}$
based on function $\textit{type}: \textit{Class} \longrightarrow \textit{Type}$
shown in Figure~\ref{fig:ir}.
When the sequence $\overline{\alpha}$ is not empty,
the class type $\mathcal{T}\prim{<}\overline{\alpha}\prim{>}: \overline{t}$
is referred to a~\emph{type constructor}.
The type system of~\ir~also features
a type instance
($N\prim{<}\overline{t}\prim{>}$)
that instantiates a type constructor $N$
with a sequence of type arguments $\overline{t}$.
\begin{figure}[t]
\scriptsize
\center
\begin{subfigure}{0.48\linewidth}
\begin{bnf*}
    \bnfprod{$\Lambda \in \textit{API}$}
    {\bnftd{$\overline{v}$}}\\[-0.5mm]
    \bnfprod{$v \in \textit{APIMember}$}
    {\bnftd{$c$} \bnfor \bnftd{$d$}}\\[-0.5mm]
    \bnfprod{$c \in \textit{Class}$}
    {\bnftd{\primsem{class} $\mathcal{C}\primsem{<}\overline{\alpha}\primsem{>}\ \primsem{extends}\ \overline{t}:\ \overline{v}$}}\\[-0.5mm]
    \bnfprod{$d \in \textit{Def}$}
    {\bnftd{\primsem{fun} $m\primsem{<}\overline{\alpha}\primsem{>}(\overline{x: t}): t$} \bnfor \bnftd{\primsem{var} $f: t$}}\\[-0.5mm]
    \bnfprod{$p \in \textit{Program}$}
    {\bnftd{$\overline{e}$}}\\[-0.5mm]
    \bnfprod{$e \in \textit{Expr}$}
    {\bnftd{$\epsilon$} \bnfor \bnftd{$\textit{constant}(t)$}
    \bnfor \bnftd{$e.f$}}\\[-0.5mm]
    \bnfmore{\bnfor \bnftd{\primsem{local var} $x: t = e$}}\\[-0.5mm]
    \bnfmore{\bnfor \bnftd{$e.m\primsem{<}\overline{t}\primsem{>}(\overline{e})$}}\\[-0.5mm]
    \bnfprod{$f, x \in {\it VariableName}$}
    \bnftd{\text{the set of field names in lib $\Lambda$}}\\[-0.5mm]
    \bnfprod{$m \in {\it MethodName}$}
    \bnftd{\text{the set of method names in lib $\Lambda$}}\\[-0.5mm]
    \bnfprod{$\mathcal{C} \in {\it ClassName}$}
    \bnftd{\text{the set of class names in lib $\Lambda$}}\\[-0.5mm]
\end{bnf*}
\vspace{-8mm}
\caption{Syntax}
\end{subfigure}
\hspace{2mm}
\begin{subfigure}{0.48\linewidth}
\begin{bnf*}
    \bnfprod{$t \in {\it Type}$}
    {\bnftd{$\top$}\bnfor\bnftd{$\bot$} \bnfor \bnftd{$\alpha$} \bnfor
    \bnftd{$N$} \bnfor
    \bnftd{$N\primsem{<}\overline{t}\primsem{>}$}}\\[-0.5mm]
    \bnfprod{$\alpha \in {\it TypeVariable}$}
    {\bnftd{$\phi: t$}}\\[-0.5mm]
    \bnfprod{$N \in \textit{ClassType}$}
    {\bnftd{$\mathcal{T}\primsem{<}\overline{\alpha}\primsem{>}$}: \overline{t}}\\[-0.5mm]
    \bnfprod{$\phi \in \textit{TypeVarName}$}
    {\bnftd{\text{the set of type variable names in lib $\Lambda$}}}\\[-0.5mm]
    \bnfprod{$\mathcal{T} \in {\it TypeName}$}
    {\bnftd{\text{the set of type names in lib $\Lambda$}}}
\end{bnf*}
\caption{Types}
\end{subfigure}
\begin{subfigure}{\linewidth}
\centering
\vspace{5mm}
\begin{minipage}{0.45\linewidth}
\centering
\begin{align*}
    &\textit{UpBound}: \textit{TypeVariable} \longrightarrow \textit{Type}\\
    &\textit{UpBound}(\phi: t) = t\\
&\textit{type}: \textit{Class} \longrightarrow \textit{Type}& \\
&\textit{type}(\primsem{class}\ \mathcal{C}\primsem{<}\overline{\alpha}\primsem{>})\ \primsem{extends}\ \overline{t}: \overline{d}) = \mathcal{C}\primsem{<}\overline{\alpha}\primsem{>}: \overline{t}\\
    &\textit{typeName}: \textit{Type} \longrightarrow \textit{TypeName}\\
    &\textit{typeName}(\mathcal{T}\primsem{<}\overline{\alpha}\primsem{>}: \overline{t}) = \mathcal{T}\\
    &\textit{typeName}(N\primsem{<}\overline{t}\primsem{>}) = \textit{typeName}(N)\\
\end{align*}
\end{minipage}
\hspace{5mm}
\begin{minipage}{0.45\linewidth}
\centering
\begin{align*}
    &\textit{className}: \textit{Class} \longrightarrow \textit{ClassName}\\
    &\textit{className}(\primsem{class }\mathcal{C}\primsem{<}\overline{\alpha}\primsem{>}\primsem{ extends }\overline{t}: \overline{v}) = \mathcal{C}\\
    &\textit{fields}: \textit{Class} \longrightarrow \textit{Def}\\
    &\textit{fields}(\primsem{class }\mathcal{C}\primsem{<}\overline{\alpha}\primsem{>}\primsem{ extends }\overline{t}: \overline{v}) = \{d\ |\ d \in \overline{v} \land d \textit{ is a field}\}\\
    &\textit{methods}: \textit{Class} \longrightarrow \textit{Def}\\
    &\textit{methods}(\primsem{class }\mathcal{C}\primsem{<}\overline{\alpha}\primsem{>}\primsem{ extends }\overline{t}: \overline{v}) = \{d\ |\ d \in \overline{v} \land d \textit{ is a function}\}\\
\end{align*}
\end{minipage}
\caption{Auxiliary functions operating on types and classes.}
\end{subfigure}
\begin{subfigure}{\linewidth}
\vspace{3mm}
\begin{mathpar}
\inferrule[field]{
    c \in \Lambda\\
    c \text{ is a class}\\
    \textit{typeName(t)} = \textit{className}(c)\\\\
    \primsem{var f}: t \in \textit{fields}(c)\\
    \langle t', \sigma\rangle = \textit{decompose}(t)
}
{
    \textit{Field}(\Lambda, t, f) = \langle \primsem{var f}: t, \sigma\rangle
}
\hva \and
\inferrule[field-inheritance]{
    c \in \Lambda\\
    c = \primsem{class }\mathcal{C}\primsem{<}\overline{\alpha}\primsem{>} \primsem{ extends }\overline{t}: \overline{v}\\\\
    \textit{typeName(t)} = \textit{className}(c)\\
    \primsem{var f}: t \not\in \textit{fields}(c) \\\\
    t' \in \overline{t}\\
    s = \textit{Field}(\Lambda, t', f)
}
{
    \textit{Field}(\Lambda, t, f) = s
}
\hva \and
\inferrule[method]{
    c \in \Lambda\\
    c \text{ is a class}\\
    \textit{typeName(t)} = \textit{className}(c)\\\\
    \primsem{fun } m\primsem{<}\overline{\alpha}\primsem{>}(\overline{x: p}): t \in \textit{methods}(c)\\\\
    \langle t', \sigma\rangle = \textit{decompose}(t)
}
{
    \textit{Method}(\Lambda, t, m) = \langle \primsem{fun } m\primsem{<}\overline{\alpha}\primsem{>}(\overline{x: p}): t , \sigma\rangle
}
\hva \and
\inferrule[method-inheritance]{
    c \in \Lambda\\
    c = \primsem{class }\mathcal{C}\primsem{<}\overline{\alpha}\primsem{>} \primsem{ extends }\overline{t}: \overline{v}\\\\
    \textit{typeName(t)} = \textit{className}(c)\\
    \primsem{fun } m\primsem{<}\overline{\alpha}\primsem{>}(\overline{x: p}): t \not\in \textit{methods}(c)\\\\
    t' \in \overline{t}\\
    s = \textit{Method}(\Lambda, t', f)
}
{
    \textit{Method}(\Lambda, t, m) = s
}
\hva \and
\end{mathpar}
\vspace{-4mm}
\caption{Lookup functions}
\end{subfigure}
\hfill
\begin{subfigure}{\linewidth}
\begin{mathpar}
\inferrule[epsilon]{
}
{
    \Lambda \vdash \epsilon: \bot
}
\hva \and
\inferrule[constant]{
}
{
    \Lambda \vdash \textit{constant(t)}: t
}
\hva \and
\inferrule[field]{
    \Lambda \vdash e: r\\
    \textit{Field}(\Lambda, r, f) = \langle\primsem{var } f: t, \sigma\rangle
}
{
    \Lambda \vdash e.f: \sigma t
}
\hva \and
\inferrule[method call]{
    \Lambda \vdash e_1: r\\
    \Lambda \vdash \overline{e_2}: \overline{p'} \\\\
    \textit{Method}(\Lambda, r, m) = \langle\primsem{fun } m\primsem{<}\overline{\alpha}\primsem{>}(\overline{x: p}): t, \sigma'\rangle \\
    \sigma = \sigma' \cup [\overline{\alpha} \mapsto \overline{t}]\\
    \overline{p'} <: \sigma \overline{p}
}
{
    \Lambda \vdash e_1.m\primsem{<}\overline{t}\primsem{>}(\overline{e_2}): \sigma t
}
\hva \and
\inferrule[local var]{
    \Lambda \vdash e: t'\\
    t' <: t
}
{
    \Lambda \vdash \primsem{local var } x: t = e: \bot
}
\hva \and
\end{mathpar}
\caption{Typing rules of~\ir.}
\end{subfigure}
\vspace{-4mm}
\caption{The syntax and the types in the~\ir.}
\label{fig:ir}
\end{figure}

Finally,
the type system of~\ir~follows
the usual typing rules
using the judgment $\Lambda \vdash e: t$
as shown in Figure~\ref{fig:ir}.
In our setting,
the given API $\Lambda$ acts as our global typing environment,
and the type of an empty expression is $\bot$
($\Lambda \vdash \epsilon: \bot$).
In the text and examples,
we use the shorthand
(1) $\phi$ for $\phi: \top$,
(2) $\mathcal{T}: \overline{t}$
for $\mathcal{T}\prim{<}\emptyset\prim{>}: \overline{t}$,
(3) $\mathcal{T}$ for $\mathcal{T}: \top$,
and (4) $\mathcal{T}\prim{<}\overline{t_2}\prim{>}$
for $(\mathcal{T}\prim{<}\overline{\alpha}\prim{>}: \overline{t_1})\prim{<}\overline{t_2}\prim{>}$.

Our language also defines
the usual type substitution and type unification operations.
A type substitution $\sigma \in \Sigma$ is a mapping
that replaces all occurrences of
a type variable $\alpha$ with a given type $t$.
We use the symbol $\sigma t$ to denote
the application of a substitution $\sigma$ on type $t$.
In this work,
type unification
(\textit{unify}: $\textit{Type} \times \textit{Type} \longrightarrow \Sigma$)
takes two types ($t_1, t_2 \in \textit{Type}$)
and identifies
a substitution $\sigma$ so that
the type $\sigma t_2$ is equal
subtype of $t_1$.
In the following,
the symbol $<:$ indicates
the subtype relation,
and $\textit{UpBound}(\alpha)$
gives the upper bound of a type variable $\alpha$
as shown in Figure~\ref{fig:ir}.

%

\begin{definition}[Validity of type substitution]
\label{def:validity-sub}
A type substitution $\sigma \in \Sigma$ is called~\emph{valid}
when $\forall \alpha \in \textit{Dom}(\sigma).\ \sigma(\alpha) <: \textit{UpBound}(\alpha)$.
\end{definition}
This definition expresses
that a type substitution is considered valid,
when every type variable in the substitution
is instantiated with a type that respects
the upper bound of the type variable.
For example,
the substitution $s = [\alpha\mapsto t_1]$
is valid when $\alpha$ has the upper bound $\top$,
because $t_1 <: \top$.
On the contrary,
the substitution is invalid when
$\alpha$ is bounded to type $t_2$,
and $t_1$ is not a subtype of $t_2$.

\begin{definition}[Subsumption]
\label{def:subsumption}
Consider two type substitutions $\sigma_1, \sigma_2 \in \Sigma$.
We say that substitution $\sigma_1$ subsumes $\sigma_2$,
denoted as $\sigma_1 \sqsubseteq \sigma_2$,
when $\forall \alpha \in \textit{Dom}(\sigma_1).\ \sigma_1(\alpha) = \sigma_2(\alpha)$.
\end{definition}
The subsumption relation holds between
two type substitutions,
when all type variables in a substitution $\sigma_1$
are instantiated with exactly the~\emph{same} type
as in another substitution $\sigma_2$.
The subsumption relation is reflexive,
and for an empty substitution $\epsilon$,
we have
$\forall \sigma \in \Sigma. \epsilon \sqsubseteq \sigma$.

\begin{definition}[Type decomposition]
\label{def:decomposition}
Type decomposition ($\textit{Type} \longrightarrow \Sigma \times \textit{Type}$)
is an operation that decomposes a given type $t_1 \in \textit{Type}$
into a substitution $\sigma$ and another type $t_2 \in \textit{Type}$,
so that $\sigma t_2 = t_1$.
It is defined as:
\begin{align*}
    \textit{decompose(t)} &= \langle[\overline{\alpha} \mapsto \overline{t_2}], \mathcal{T}\prim{<}\overline{\alpha}\prim{>}: \overline{t_1}\rangle & \text{if } t = (\mathcal{T}\prim{<}\overline{\alpha}\prim{>}: \overline{t_1})\prim{<}\overline{t_2}\prim{>}\\
    \textit{decompose}(t) &= \langle\epsilon, t\rangle & \text{otherwise}
\end{align*}
\end{definition}
In essence,
type decomposition allows us to decompose
a type instance $N\prim{<}\overline{t}\prim{>}$ into
(1) the type constructor $N$,
and (2) the type substitution
that replaces all formal type variables in $N$
with the provided type arguments $\overline{t}$.
For example,
consider the type constructor $N = \mathcal{T}\prim{<}\alpha\prim{>}: \top$
and its type instance $t = N\prim{<}t_2\prim{>}$.
In this example,
$\textit{decompose}(t)$ returns $\langle[\alpha \mapsto t_2], N\rangle$.
For a non polymorphic type $t$,
type decomposition simply returns an empty substitution
and the input type $t$.

\subsection{API Graph}
\label{sec:api-graph}

We define an API directed graph~as $G = (V, E)$,
where $V$ is the set of nodes
corresponding to either
a type $t \in \textit{Type}$
or an API definition $d \in \textit{Def}$
(i.e., a method or a field),
$E \subseteq V \times V \times L$
is the set of edges
whereas $L = \Sigma$
is the set of edge labels (representing
the set of valid type substitutions).
To construct the~\graph~
we examine a given API $\Lambda\in \textit{API}$
and proceed as follows.
\begin{itemize}
    \item Iterate over the set of classes in topological order with regards to
    their inheritance chain.
    Convert every class $c$ into a type $t$
    based on function $\textit{type}$
    defined in Figure~\ref{fig:ir},
    and add the resulting type $t$ into the graph.
    Then,
    we iterate over each member (i.e., a function or a field)
    $d \in \textit{Def}$\ belonging to class $c$ and proceed as follows.
\item Add node $d$ to the~\graph.
\item Add edge $t \overset{\epsilon}{\rightarrow} d$ to the~\graph,
    if $d$ is an~\emph{instance} method or a field of the class $c$
    and $t = \textit{type(c)}$.
\item Add edge $d \overset{\sigma}{\rightarrow} r'$ to the~\graph,
    if $r$ is the return type of $d$ and
    $\textit{decompose}(r) = \langle\sigma, r'\rangle$
    according to Definition~\ref{def:decomposition}.
\end{itemize}

\begin{figure}[t]
\begin{subfigure}{0.33\textwidth}
\centering
\begin{lstlisting}[language=java, basicstyle=\ttfamily\tiny]
class Utils {
  static <X> List<X> createList();
}
class List<T> {
  List(int size);
  boolean add(T elem);
  Set<T> toSet();
}
class Set<E> {
  Set(int size);
  boolean add(E elem);
  List<E> toList();
}
\end{lstlisting}
\end{subfigure}
\begin{subfigure}{.66\linewidth}
\centering
\includegraphics[scale=0.32]{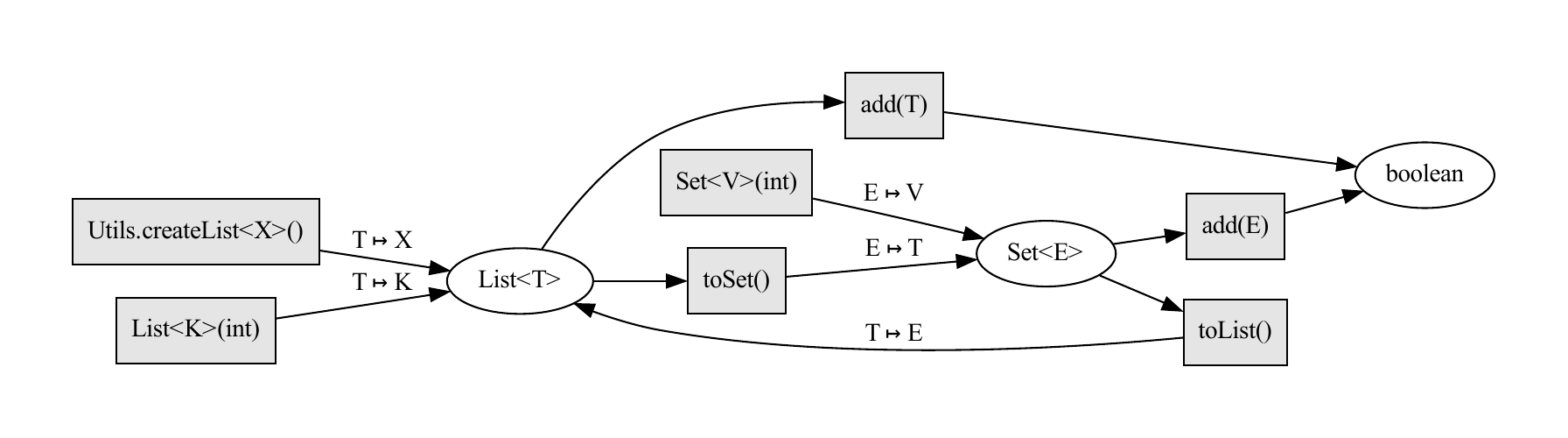}
\end{subfigure}
\vspace{-4mm}
\caption{An example API written in Groovy and its~\graph.
Oval nodes represent types,
while rectangular nodes denote definitions (e.g., methods).}
\label{fig:example-api}
\vspace{-4mm}
\end{figure}

Conceptually,
the set of edges determines the following
relationships.
The edge $t \overset{\epsilon}{\rightarrow} d$ indicates
that the definition $d$
is applied to a value of type $t$ (receiver type).
The edge $d \overset{\sigma}{\rightarrow} r'$ denotes
that the application of the API definition $d$
returns the type given by $\sigma r'$.
For example,
consider a method $m$
whose return type is {\tt List<String>}.
In this scenario,
we add the following edge to the~\graph:
$m \overset{\alpha \mapsto \texttt{String}}{\rightarrow} \texttt{List}\prim{<}\alpha\prim{>}$,
where the target node
of the edge is the type constructor {\tt List}.
In a similar manner,
when the return type of method $m$
is a non-polymorphic type,
such as {\tt Int},
$\textit{decopose(\texttt{Int})}$ yields $\langle \epsilon, \texttt{Int}\rangle$.
Therefore,
we add the edge $m \overset{\epsilon}{\rightarrow} \texttt{Int}$
to the graph.
In our examples,
we use the shorthand $v_1 \rightarrow v_2$
for $v_1 \overset{\epsilon}{\rightarrow} v_2$.

The~\graph\ is directly inspired by
the {\it signature graph} introduced
in the program synthesis work of~\citet{jungloid}.
In a similar manner to our~\graph,
the signature graph encodes
API components as unary functions
taking one input type
(receiver type),
and producing an output type (return type).
The intuition behind this
is the identification of chains of method calls/field accesses
(e.g., {\tt x.f1.m1().f2...})
that lead to a specific type $t_{\textit{out}}$.
This can be achieved by
simply querying the graph for paths
between an input type $t_{\textit{in}}$
and the target type $t_{\textit{out}}$
using standard graph reachability algorithms,
such as Dijkstra's algorithm.

However,
the presence of polymorphic components complicates
this straightforward approach.
Our~\graph\ addresses this
by refining the signature graph:
we label edges with type substitutions,
allowing us to handle parametric polymorphism
without compromising the size of the~\graph. 
Our detailed approach for identifying chains
of method calls/field accesses 
in the presence of parametric polymorphism
is explained later in Section~\ref{sec:concretization}.

\point{Example}
Figure~\ref{fig:example-api} shows an
example API (written in Groovy) and its~\graph.
The~\graph~contains seven definitions
(i.e., methods) depicted with gray color.
Three definitions have no incoming edges,
as their application does not require
a receiver object.
Two of them
(i.e., $\texttt{List<K>(int)}$,
$\texttt{Set<V>(int)}$)
stand for the constructors
of the generic classes {\tt List} and {\tt Set}
respectively (lines 4, 9),
while one of them is for the
static polymorphic method of class {\tt Utils}.
Finally,
the example~\graph~includes three types
represented by oval nodes,
two of which denote the type constructors
{\tt List} and {\tt Set},
while one type node corresponds to {\tt boolean}.


\subsection{API Enumeration Problem Formulation}
\label{sec:api-enumeration}

Having presented the notion of~\graph~(Section~\ref{sec:api-graph}),
we now formulate the problem of API enumeration.
API enumeration systematically explores
all the unique typing combinations
that can be used to
invoke a particular API component,
such as a function or a field.
The intuition is that
invoking an API component
through different typing patterns
lets us exercise the implementation
of many type-related operations in the compiler,
including,
subtyping rules,
or method resolution.
In what follows,
we use an~\graph~as our typing environment.
In fact,
an~\graph~can help type expressions
that use API components,
because the type signatures of
every API entity is included in $G$.
For example,
hereafter,
$G \vdash e.f: t$ means
that the type of the field access $e.f$ is $t$.
This typing process is 
based on the type signature of field $f$
found in $G$.

We first introduce~\emph{abstract typed expressions},
an abstraction over the domain of expressions defined
in~\ir~(Figure~\ref{fig:ir}).
An abstract typed expression is given by:
\begin{bnf*}
    \bnfprod{$\hat{e} \in \hat{\textit{Expr}}$}
    {\bnftd{$[t]$} \bnfor \bnftd{$[t].f$}
    \bnfor \bnftd{$[t_1].m\prim{<}\overline{t}\prim{>}(\overline{[t_2]})$}
    \bnfor \bnftd{$\prim{local var}\ x: t = \hat{e}$}}
\end{bnf*}
The notation $[t]$ represents
an~\emph{inhabitant} of type $t$~\cite{type-inhabitant},
that is
any concrete expression $e \in \textit{Expr}$
whose type is $t$.
An abstract typed expression hides
the contents and the value of a concrete expression $e$,
and considers only the type of $e$.
For example,
the abstract expression $[t].f$
indicates a field $f$ accessed
through~\emph{any} expression of type $t$.
We encode every abstract typed expression $\hat{e}$
using \emph{typing sequences}.
A typing sequence succinctly captures
the types found within the holes
of abstract typed expressions.
A typing sequence is denoted by the symbol $\llbracket\rrbracket$:
\begin{align*}
    \llbracket[t]\rrbracket &\rightarrow \langle t\rangle\\
    \llbracket[t].f\rrbracket &\rightarrow \langle t, \bot\rangle\\
    \llbracket[t_1].m\prim{<}\overline{t}\prim{>}(\overline{[t_2]})\rrbracket &\rightarrow \langle t_1, \overline{t_2}\rangle \\
    \llbracket\prim{local var}\ x: t = \hat{e}\rrbracket &\rightarrow \llbracket\hat{e}\rrbracket \boldsymbol{\cdot} t
\end{align*}
In the preceding rules,
the symbol $\boldsymbol{\cdot}$ means
appending an element to the end of a sequence.
A singleton sequence consisting of type $t$
represents an inhabitant of $t$.
When a typing sequence contains more than two types,
the first element of the sequence stands for
the receiver type of a field access/method call,
and the remaining elements correspond to the
parameter types of the application (if any).
Finally,
in the case of local variable definitions,
the final element of the typing sequence
is the expected type of
the entire abstract expression
found on the right-hand side.

We define a concretization function $\gamma$
that allows us to map a typing sequence
and an API definition $d$
into a set of concrete expressions written in~\ir\
under a given type substitution $\sigma$.

\begin{definition}
\label{def:concrete}
Let $\gamma: G, \textit{Def}, \textit{Type} \times \dots \times \textit{Type} \times \Sigma \longrightarrow \mathcal{P}(\textit{Expr})$ such that:
\begin{align*}
    \gamma(G, \bot, \langle t \rangle, \sigma) &= \{e \ |\ G \vdash e: t\} \\
    \gamma(G, \prim{var}\ f: t, \langle t, \bot\rangle, \sigma) &=
        \{e.f\ |\ G \vdash e: t\} \\
\gamma(G, \prim{fun}\ m\prim{<}\overline{\alpha}\prim{>}(\overline{x: p}): t, \langle r, \overline{p'}\rangle, \sigma) &=
 \left\{ e_1.m\prim{<}\overline{t}\prim{>}(\overline{e_2})\ \middle| \ 
        \begin{aligned}
            & G \vdash e_1 : r, G \vdash \overline{e_2}: \overline{p'}, \\
            & \overline{t} = (\sigma(\alpha_i))_{i = 1}^{n}
        \end{aligned}
\right\}\\
    \gamma(G, d, s \boldsymbol{\cdot} t, \sigma) &= \{\prim{local var}\ x: t = e\ |\ e \in \gamma(G, d, s, \sigma)\}\\
\end{align*}
\end{definition}
\noindent
The function $\gamma$
concretizes abstract typing sequences and definitions
into specific expressions that are typed under
a given~\graph\ and a type substitution.
When no API entity is provided to $\gamma$
(i.e., its second parameter is $\bot$),
the function returns
all the inhabitants of type $t$.
This translates to every expression
$e \in \textit{Expr}$
that satisfies $G \vdash e: t$.
Conversely,
if an API component $d$ is provided,
$\gamma$ yields
all the possible expressions that
invoke the definition $d$ with respect to the given typing sequence $s$.
For example,
consider the field $\prim{var}\ f: t$
and the typing sequence $s = \langle r, \bot\rangle$.
In this case,
the function $\gamma$ gives
all accesses of field $f$
via every inhabitant of type $r$.
When encountering a polymorphic function,
$\gamma$ maps every formal type variable of the function
into actual type arguments attached to the resulting method calls
based on the provided type substitution,
in particular $(\sigma(\alpha_i))_{i=1}^{n}$.
Later,
in Section~\ref{sec:concretization},
we present an under-approximation of $\gamma$
called $\hat{\gamma}$,
introduced to address the practical concerns
associated with exhaustively enumerating
all potential expressions in $\gamma$.

\begin{definition}[API typing sequence]
\label{def:typing-seq}
Given an API component $d \in \textit{Def}$,
the sequence $s_d$ is
called an~\emph{API typing sequence} of $d$
when there is \emph{an abstract typed expression} $\hat{e}$,
such that (1) $\llbracket \texttt{local var}\ x: t = \hat{e} \rrbracket = s_d$,
and (2) the abstract expression $\hat{e}$ invokes the component $d$.
We say that an expression $e \in \textit{Expr}$~\emph{realizes}
the API typing sequence $s_d$ under the substitution $\sigma$,
if $e \in \gamma(G, d, s_d, \sigma)$.
\end{definition}

\vspace{-2mm}
Based on Definition~\ref{def:typing-seq},
an API typing sequence $s_d$
reveals two key details:
First,
\emph{how} a set of types
are combined together
to invoke and use a certain API component $d$
(represented by all elements of $s_d$ except the last).
Second,
\emph{what} is the expected type
that derives from the usage of $d$
(indicated by the last element of $s_d$).

\point{Example}
Consider two typing sequences
for the method $d = \texttt{add(T)}$
defined in Figure~\ref{fig:example-api} (line~6):
$s_{d_1} = \langle \bot, \texttt{int}, \texttt{boolean} \rangle$
and $s_{d_2} = \langle \texttt{List<Int>}, \texttt{int}, \texttt{boolean} \rangle$.
Listing~\ref{lst:example-listing} contains three expressions
that realize $s_{d_1}$ and $s_{d_2}$.
The first expression (although type incorrect)
realizes $s_{d_1}$, because the first element of $s_{d_1}$
(denoted as $s_{d_1}\downarrow_1$) is $\bot$,
which represents the absence of receiver.
The last two expressions of the listing realize the
same API typing sequence $s_{d_2}$,
because the type of both {\tt new List<Int>(10)}
and {\tt Utils.<Int>createList()} is
$s_{d_2}\downarrow_1 = \texttt{List<Int>}$.
\begin{lstlisting}[language=java, mathescape=true, caption={Expressions that realize typing sequences of method {\tt List.add(T)}.}, label={lst:example-listing}]
boolean x = add(1);
boolean y = Utils.<Int>createList().add(1);
boolean z = new List<Int>(10).add(1);
\end{lstlisting}

\begin{definition}[API signature]
API signature ($G \times \textit{Def} \longrightarrow 
\textit{Type} \times \dots \times
\textit{Type}$) is a function
that maps an~\graph~$G = (V, E)$
and one API definition $d \in V$
to a typing sequence as follows:
\begin{align*}
    \textit{sig(G, \texttt{var} x: t)} &=  \langle r, \bot, t\rangle& \text{if } r\overset{l}{\rightarrow} d \in E\\
    \textit{sig(G, \texttt{var} x: t)} &=  \langle \bot, \bot, t\rangle& \text{if } d\ \text{has no incoming edges in } G\\
\textit{sig}(G, \texttt{fun}\ m\prim{<}\overline{\alpha}\prim{>}(\overline{x: p}): t) &=  \langle r, \overline{p}, t\rangle & \text{if } r\overset{l}{\rightarrow} d \in E\\
\textit{sig}(G, \texttt{fun}\ m\prim{<}\overline{\alpha}\prim{>}(\overline{x: p}): t) &=  \langle \bot, \overline{p}, t\rangle & \text{if } d\ \text{has no incoming edges in } G\\
\end{align*}
\end{definition}
\noindent
For example,
the API signature of the {\tt Utils.createList()} method
of Figure~\ref{fig:example-api} (line 2) is
$\langle\bot,\bot, \texttt{List}\prim{<}X\prim{>}\rangle$.
The method is static
and takes no parameters.
That is why
the first element (receiver type)
and the second element (parameter type)
of the resulting sequence are $\bot$.

\begin{definition}[Well-typed API typing sequence]
\label{def:well-typed}
Given an~\graph~$G$,
a definition $d \in \textit{Def}$,
and a type substitution $\sigma \in \Sigma$,
we say that an API typing sequence 
$s_d = \langle r, \overline{p}, t\rangle$
is~\emph{well-typed}
under $G$ and $\sigma$,
if $\textit{sig(G, d)} = (r', \overline{p'}, t')$
and $r <: \sigma r', \overline{p} <: \sigma \overline{p'}$,
and $t >: \sigma t'$.
\end{definition}
In essence,
this definition captures precisely the notion that 
given an API component $d$,
a typing sequence $s_d$ is well-typed under $\sigma$,
if $s_d$ describes a usage of $d$ with:
(1) a receiver whose type $r$
is a subtype of the
formal receiver type
extracted from the signature of $d$
($r <: \sigma r'$),
and (2) arguments (if present) whose types are subtypes
of the formal parameter types of $d$.
Finally,
the expected type of $d$'s application
should be any supertype of $d$'s formal return type.

\begin{theorem}
\label{prop:well-typed}
Consider an API component $d \in \textit{Def}$,
one well-typed typing sequence $s_d$ of $d$,
and a substitution $\sigma$.
The programs derived from $s_d$,
that is $\gamma(G, d, s_d, \sigma)$,
are well-typed.
\end{theorem}

\begin{proof}
    This follows straightforwardly from Definition~\ref{def:concrete}
    and the typing rules of~\ir.
\end{proof}

\point{Example}
Consider again the method $d = \texttt{add}(T)$
defined in the {\tt List} class
of our example API
(Figure~\ref{fig:example-api}, line 6).
Also,
consider a substitution $\sigma = [T \mapsto {\tt Int}]$.
The API typing sequence
$s_d =\langle \texttt{List<Int>}, \texttt{Int}, \texttt{boolean}\rangle$
is well-typed under $\sigma$,
because
(1) $\texttt{List<Int>} <: [T \mapsto \texttt{Int}]\texttt{List}\prim{<}T\prim{>}$,
(2) $\texttt{Int} <: [T \mapsto \texttt{Int}] T$,
and (3) $\texttt{boolean} >: [T \mapsto \texttt{Int}] \texttt{boolean}$.
Hence,
all the concrete expressions that come from
$s_d$,
such as the last two expressions of Listing~\ref{lst:example-listing},
are well-typed.

\begin{definition}[API enumeration]
Consider an~\graph~$G$
and an API component $d$ included in $G$.
For every type substitution $\sigma \in \Sigma$,
API enumeration computes exhaustively
a set of~\emph{well-typed} API typing sequences $S_d$
under $\sigma$
such that for every $s_d \in S_d$,
we obtain a program $p \in \gamma(G, d, s_d, \sigma)$.
\end{definition}

The algorithm of
API enumeration first enumerates~\emph{all} valid
type substitutions that instantiate the type variables
of an API definition.
Given that in the presence of polymorphic types,
substitutions are infinite
(consider {\tt List<List<List<\dots{>}{>}{>}}),
API enumeration can be set to produce
all type substitutions up to a certain depth.
This is done by instantiating each type variable
with every concrete type found in a given API,
including type instances. 
For example,
when it comes to enumerating all type substitutions of length two,
the result set includes a type of the form {\tt List<List<Int{>}{>}},
but not {\tt List<List<List<Int{>}{>}{>}}
(as the depth in this case is three).
Based on a substitution $\sigma$
and a library component $d$
with signature $\langle r, p_1\dots p_n, t\rangle$,
API enumeration generates
the Cartesian product of
sets $R, P_1,\dots P_n, T$,
where $R$ is the set consisting of subtypes of $\sigma r$,
$P_i$ is the set consisting of subtypes of $\sigma p_i$
(with $1 \le i \le n$),
and $T$ is the set containing the supertypes of $\sigma t$.
API enumeration is exponential in terms of
the number of types included in the signature of $d$.
For example,
consider a substitution $\sigma$
and an API signature that specifies four types:
$\langle r, p_1, p_2, t\rangle$.
Each of $\sigma r$,
$\sigma p_1$,
and $\sigma p_2$ contains ten subtypes,
while there are also ten supertypes of $\sigma t$.
API enumeration gives $10^4 = \nnum{10000}$
well-typed typing sequences under $\sigma$.

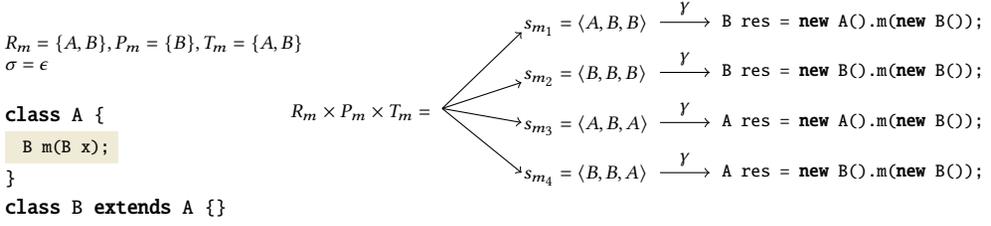
\begin{figure}
\scriptsize
\begin{tikzpicture}[x=0.75pt,y=0.75pt,yscale=-1,xscale=1]
\draw (8,110) node [anchor=north west][inner sep=0.75pt]   [align=left]
{
$R_m = \{A, B\}, P_{m} = \{B\}, T_m = \{A, B\}$\\
$\sigma = \epsilon$\\\\
\begin{lstlisting}[language=java, numbers=none, escapeinside={<@}{@>}]
class A {
<@\colorbox{eggshell}{\scriptsize\texttt{  B m(B x);}}@>
}
class B extends A {}
\end{lstlisting}
};
\draw (190,150) node [align=left]
{
$R_m \times P_{m} \times T_m = $
};
\draw (270,100) node [anchor=north west][inner sep=0.75pt]   [align=right]
{$s_{m_1} = \langle A, B, B\rangle$};
\draw (270,125) node [anchor=north west][inner sep=0.75pt]   [align=right] 
{$s_{m_2} = \langle B, B, B\rangle$};
\draw (270,150) node [anchor=north west][inner sep=0.75pt]   [align=right]
{$s_{m_3} = \langle A, B, A\rangle$};
\draw (270,175) node [anchor=north west][inner sep=0.75pt]   [align=right]
{$s_{m_4} = \langle B, B, A\rangle$};
\draw (370,95) node [anchor=north west][inner sep=0.75pt]   [align=right]
{
\begin{lstlisting}[language=java, numbers=none, basicstyle=\ttfamily\scriptsize]
B res = new A().m(new B());
\end{lstlisting}
};
;
\draw (370,120) node [anchor=north west][inner sep=0.75pt]   [align=right]
{
\begin{lstlisting}[language=java, numbers=none, basicstyle=\ttfamily\scriptsize]
B res = new B().m(new B());
\end{lstlisting}
};
\draw (370,145) node [anchor=north west][inner sep=0.75pt]   [align=right]
{
\begin{lstlisting}[language=java, numbers=none, basicstyle=\ttfamily\scriptsize]
A res = new A().m(new B());
\end{lstlisting}
};
\draw (370,170) node [anchor=north west][inner sep=0.75pt]   [align=right]
{
\begin{lstlisting}[language=java, numbers=none, basicstyle=\ttfamily\scriptsize]
A res = new B().m(new B());
\end{lstlisting}
};
\draw[->] (230, 148) -- (270, 110);
\draw[->] (230, 148) -- (270, 135);
\draw[->] (230, 148) -- (270, 155);
\draw[->] (230, 148) -- (270, 180);
\draw[->] (340, 105) -- (365, 105) node[midway, above] {$\gamma$};
\draw[->] (340, 130) -- (365, 130) node[midway, above] {$\gamma$};
\draw[->] (340, 155) -- (365, 155) node[midway, above] {$\gamma$};
\draw[->] (340, 180) -- (365, 180) node[midway, above] {$\gamma$};
\end{tikzpicture}
\caption{Enumerating well-typed API typing sequences
for method {\tt B m(B x)} using substitution $\sigma = \epsilon$.
The set $R_m$ contains the subtypes of the formal receiver type of {\tt m},
$P_{m}$ contains the subtypes of the first formal parameter type of {\tt m},
while $T_m$ contains the supertypes of the return type of {\tt m}.
Using the function $\gamma$,
each typing sequence leads to programs
that invoke method {\tt m} via a unique typing pattern.
This,
in turn,
triggers various type-related operations in the compiler.
}
\label{fig:api-enumeration-ex}
\vspace{-3mm}
\end{figure}

\point{Example}
Figure~\ref{fig:api-enumeration-ex} illustrates
the concept of API enumeration for
the non-polymorphic method {\tt B m(B x)}.
Based on the signature of method {\tt m},
we compute the sets $R_m$, $P_{m}$,
and $T_m$,
and then we take their Cartesian product.
Each element of this product
corresponds to a well-typed API typing sequence,
which eventually leads to concrete programs
that invoke method {\tt m}
through a unique typing combination.
For example,
the type of the receiver in the first program
is {\tt A},
while the type of the receiver in the second program
is type {\tt B}.
API enumeration for polymorphic components
follows a similar process
by enumerating well-typed API sequences
under every valid type substitution.

\point{Enumerating ill-typed API typing sequences}
An API typing sequence $s_d$ of
a definition $d$
is considered ill-typed
when at least one enclosing type of the sequence is~\emph{incompatible}
with the signature of $d$.
Consider an~\graph~$G$,
a definition $d$,
and a substitution $\sigma$.
An API typing sequence $s_d = \langle r,\overline{p},t \rangle$
is ill-typed under $G$ and $\sigma$,
if $\textit{sig}(G, d) = \langle r', \overline{p'}, t' \rangle$
and $r \not{<:} \sigma r'$ and $r \not{>:} \sigma r'$,
or $\overline{p} \not{<:} \sigma \overline{p'}$,
or $t \not{>:} \sigma t'$.
Notice that a receiver type
is incompatible when
it is neither a supertype nor a subtype
of the formal receiver type.
This happens
to avoid constructing a well-typed sequence
due to inheritance.
In a similar manner to Theorem~\ref{prop:well-typed},
the programs derived from ill-typed API typing sequences
are ill-typed.
Consider again the method $d = \texttt{List.add(T)}$
of Figure~\ref{fig:example-api}.
Under the substitution $\sigma = [T \mapsto {\tt Int}]$,
the sequence
$s_d =\langle \texttt{List<Int>}, \texttt{String}, \texttt{boolean}\rangle$
is ill-typed,
because $\texttt{String} \not{<:} [T \mapsto \texttt{Int}] T$.
An expression that realizes this sequence
is deemed ill-typed:
\begin{lstlisting}[language=Java,numbers=none]
boolean x = Utils.<Int>createList("str value") // String is incompatible to Int
\end{lstlisting}

\subsection{Concretization of API Typing Sequences}
\label{sec:concretization}

Every API typing sequence that
arises from API enumeration
is converted into a concrete test program.
We present a method that~\emph{under-approximates}
the concretization function $\gamma$
presented in Definition~\ref{def:concrete}.
To identify type inhabitants,
our method examines the~\graph~and
enumerates all paths that
reach a certain type $t$.
Each path represents an inhabitant
consisting of a chain of method calls
or field accesses.
In the presence of parametric polymorphism,
some paths are infeasible,
as the underlying type variables found
in these paths are instantiated with
incompatible types
that do not ultimately lead to a given type $t$.
We explain the details on how
we determine the feasibility of each path
and properly instantiate the corresponding type variables.

\begin{figure}
\begin{minipage}{.43\textwidth}
\begin{lstlisting}[language=java]
class Utils {
  static <X, Y> Map<X, Y> mapOf();
  static Map<String, String> mapOfStrs();
}
class Map<K, V> {
  Set<K> keySet();
}
class Set<E> {}
\end{lstlisting}
\end{minipage}
\begin{minipage}{.56\textwidth}
\center
\includegraphics[scale=0.36]{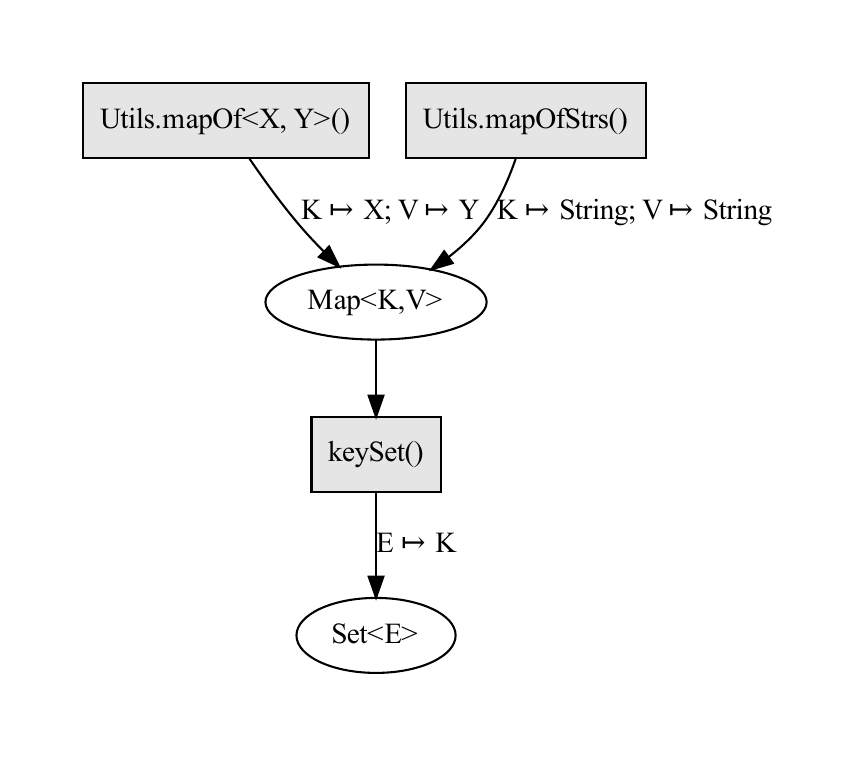}
\end{minipage}
\vspace{-5mm}
\caption{An API relying on parametric polymorphism and its~\graph.}
\label{fig:ex-with-param}
\vspace{-3mm}
\end{figure}

\point{Finding type inhabitants via API graph reachability}
In the simplest scenario,
where parametric polymorphism is not present,
finding inhabitants of a type $t$ is given by
standard graph reachability algorithms
that compute the set of paths that reach node $t$~\cite{jungloid}.
However,
parametric polymorphism complicates the process,
making the simple
path search yield incorrect results.
To illustrate this,
consider the API and its corresponding~\graph~of
Figure~\ref{fig:ex-with-param}.
Now assume that we want to find an inhabitant
of polymorphic type {\tt Set<Int>}.
The na\"ive path search
identifies two potential paths
that reach the type constructor {\tt Set}:
path (1) and path (2).
The expression {\tt mapOfStrs().keySet()},
which originates from path (1)
${\tt Utils.mapOfStrs()} \rightarrow {\tt Map}\prim{<}K, V\prim{>}
\rightarrow {\tt keySet()} \rightarrow {\tt Set}\prim{<}E\prim{>}$,
leads to an incompatible type {\tt Set<String>}.
In contrast,
path (2)
${\tt Utils.mapOf}\prim{<}X,Y\prim{>}{\tt ()} \rightarrow {\tt Map}\prim{<}K,V\prim{>}
\rightarrow {\tt keySet()} \rightarrow {\tt Set}\prim{<}E\prim{>}$
involves a polymorphic function
({\tt Utils.mapOf}).
The expression from path (2) is correctly
identified as
an inhabitant of {\tt Set<Int>}
\emph{only} if the type variable $X$
defined in {\tt Utils.mapOf}
is instantiated with type {\tt Int}.
To identify infeasible paths and properly
instantiate type variables of polymorphic components,
we now present a refined path search approach
that deals with parametric polymorphism.

Every path in an~\graph~may consist of
polymorphic API components.
Every polymorphic component introduces
some type variables.
In the above example,
path (1) introduces type variables $K, V$
and $E$.
Similarly,
path (2) introduces type variables,
$X, Y, K, V$, and $E$.
The labels 
(aka substitutions---recall Section~\ref{sec:api-graph})
found on top of each edge indicate
instantiations of the corresponding type variables.
However,
a path may also
contain free type variables
whose instantiation is not given by such substitutions.
For example,
path (1) of Figure~\ref{fig:ex-with-param}
contains no free type variables,
as both type variables $K$ and $V$
are instantiated with {\tt String}
while type variable $E$ is instantiated with
type $K$.
On the other hand,
path (2) involves two free type variables,
namely, $X$ and $Y$.
These type variables do not participate
in the left-hand side of a substitution.
From now onwards,
given a path $p$,
we use the notation $\textit{TypeVar(p)}$
to obtain the set of type variables of path $p$.
To take the set of free type variables
of a path $p$,
we use the symbol $\textit{FV}(p)$.

\begin{definition}
Let a type $t$ and its decomposition
$\textit{decompose}(t) = \langle\sigma, t'\rangle$.
Given an~\graph~$G$,
we say that a path $p$~\emph{forms}
an inhabitant of type $t$,
when
(1) the path $p$ leads to target node $t'$,
and (2) there is a
\emph{a valid} substitution $\sigma'$,
such that $\forall\alpha\in\textit{TypeVar(p)}.\ \alpha \in \textit{Dom}(\sigma')$
and $\sigma \sqsubseteq \sigma'$.
Our goal is then to find the set of
paths $P$
so that every path $p \in P$
forms an inhabitant of type $t$.
\end{definition}
The definition above summarizes the problem
of path search,
which aims to find type inhabitants even
in the presence of polymorphic types.
Given a target type $t$,
we initially decompose it
into a substitution $\sigma$ and a type $t'$
according to Definition~\ref{def:decomposition}.
Then,
we search the~\graph~$G$ to find the set of paths
that reach node $t'$.
A path $p$ forms a type inhabitant
if there is a~\emph{valid} substitution $\sigma'$
(recall Definition~\ref{def:validity-sub})
that includes every type variable in $\textit{TypeVar(p)}$,
and subsumes the original substitution $\sigma$
according to Definition~\ref{def:subsumption}
($\sigma \sqsubseteq \sigma'$).
Conceptually,
the substitution $\sigma$
that we obtain after performing the type decomposition on $t$,
\emph{constrains} the instantiation
of some type variables found in path $p$.
Therefore,
we need to ensure
that the substitution $\sigma'$ contains
compatible assignments
for the type variables included in $\sigma$.
If no such valid substitution $\sigma'$ exists,
then we consider the path as~\emph{infeasible}:
it cannot yield an
expression of the target type $t$.

\point{Example}
Consider again the~\graph~of Figure~\ref{fig:ex-with-param}
and the problem of finding inhabitants of type {\tt Set<Int>}.
This type is decomposed into the
substitution $\sigma = [E \mapsto {\tt Int}]$
and the type constructor {\tt Set}\prim{<}E\prim{>}.
Again,
there are two paths
that reach the type constructor {\tt Set}.
Path (1) is infeasible
as the only substitution $\sigma$
that stems from this path
is $\sigma' = [K \mapsto {\tt String}, V\mapsto {\tt String}, E \mapsto {\tt String}$],
which is incompatible with $\sigma$,
because $\sigma(E) \neq \sigma'(E)$.
On the other hand,
path (2) is feasible
because there exists a substitution $\sigma'$
such that $\sigma \sqsubseteq \sigma'$.
Specifically,
we have $\sigma' = [X \mapsto {\tt Int}, Y\mapsto {\tt String}, K \mapsto {\tt Int}, Y \mapsto {\tt String}, E \mapsto {\tt Int}]$,
as
$\sigma(E) = \sigma'(E) = {\tt Int}$.

\setlength{\textfloatsep}{1pt}
\begin{algorithm}[t]
\small
\SetKwProg{Pn}{fun}{=}{}
\SetKwFunction{PATH}{find\_API\_paths}
\DontPrintSemicolon
\SetKwInOut{Input}{Input}
\SetKwInOut{Output}{Output}
\SetKwFor{Case}{case}{$\Rightarrow$}{}%
\SetKwFor{Switch}{match}{with}{}%
\SetInd{0.1em}{1em}
\Pn{\PATH{$G, t$}}{
    $\langle\sigma,t'\rangle \gets \textit{decompose(t)}$\;
    $\textit{paths} \gets \{\langle\langle t \rangle, \bot\rangle\}$\;
    \For{$p \in \textit{Paths}(G, t')$}{
        $\textit{sub} \gets \textit{gather all substitutions of path p}$\;
        $\sigma' \gets \textit{perform constant propagation on } \sigma\ \textit{based on } \textit{subs}$\;
        $t_2 \gets \sigma't'$\;
        $\sigma_1 \gets \textit{unify}(t, t_2)$\;
        \lIf{$\sigma_1 = \epsilon$}{
            \Continue
        }
        $\sigma_2 \gets \forall\alpha\in\textit{FV}(p).\ \textit{instantiate } \alpha\: \: \:  \textit{if } \alpha \not\in \sigma_1$\;
        $\sigma' \gets \sigma_1 \cup \sigma_2$\;
        $\textit{paths} \gets \textit{paths} \cup \{\langle p, \sigma'\rangle\}$\;
    }
    \Return{$\textit{paths}$}
}
\caption{Algorithm for finding paths that form inhabitants of type $t$}
\label{alg:find-paths}
\end{algorithm}

\begin{figure}[t]
\centering
\begin{subfigure}{0.45\linewidth}
\includegraphics[scale=0.32]{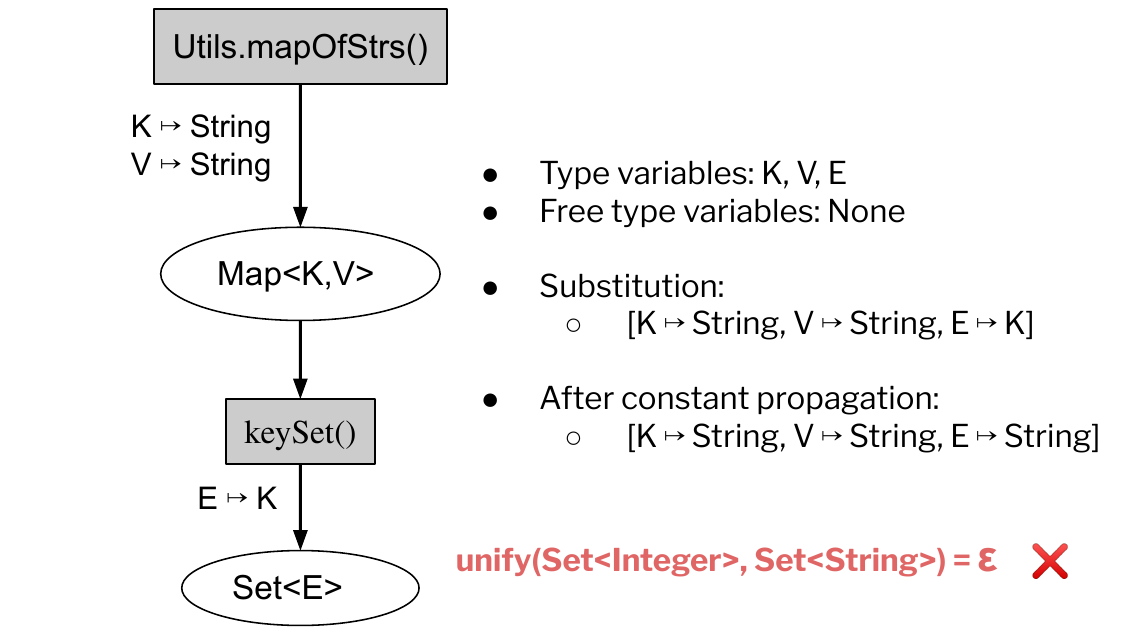}
\caption{An infeasible path.}
\label{fig:first-path}
\end{subfigure}
~
\hspace{3mm}
\begin{subfigure}{0.54\linewidth}
\includegraphics[scale=0.32]{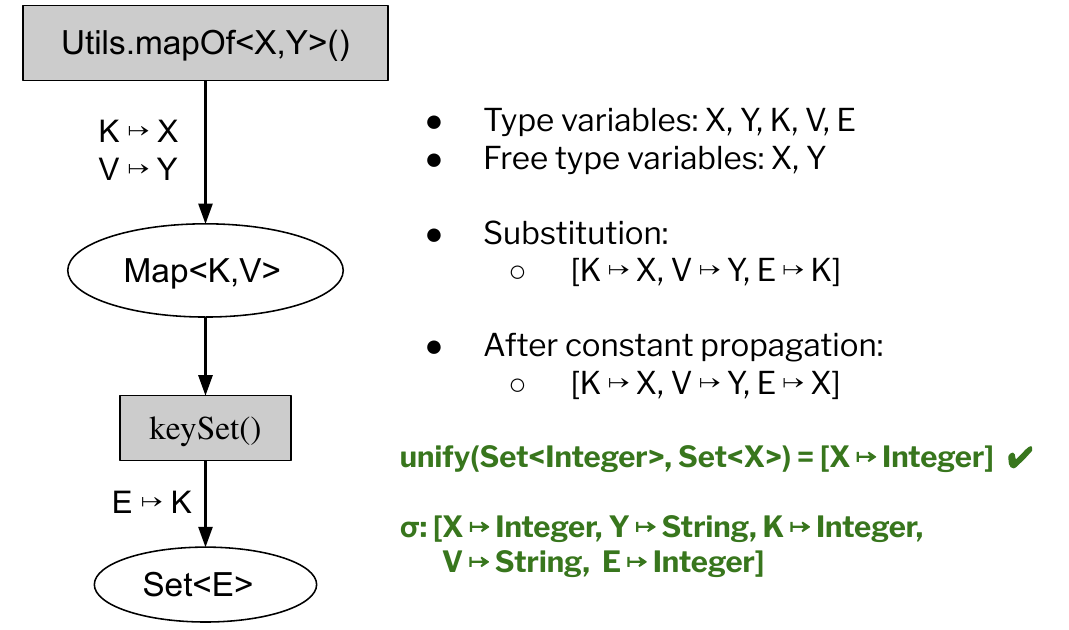}
\caption{An inhabitant of {\tt Set<Int>}.}
\label{fig:second-path}
\end{subfigure}
\vspace{-3mm}
\caption{Determining which of the paths
in the~\graph~of Figure~\ref{fig:ex-with-param}
form an inhabitant of type {\tt Set<Int>}.}
\label{fig:individ-paths}
\vspace{2mm}
\end{figure}

\point{Algorithm}
Algorithm~\ref{alg:find-paths} outlines
our method for identifying paths
that form inhabitants of a given type.
We describe the algorithm
in the context of
finding inhabitants of {\tt Set<Int>}
based on
the~\graph~shown in Figure~\ref{fig:ex-with-param}.
The algorithm takes as input
an~\graph~$G$ and a type $t$,
and returns a set of pairs,
where each pair consists of
a path and a substitution.
The algorithm starts with
decomposing the given type $t$
into a substitution $\sigma$
and another type $t'$.
Then,
it enumerates all~\emph{acyclic} paths that reach node $t'$.
In our example,
the are two paths
that reach the type constructor {\tt Set}---each of those paths is shown individually
in Figure~\ref{fig:individ-paths}.

It is now time to determine
which of the available paths
are feasible.
For every path $p$,
the algorithm examines and gathers all substitutions
(i.e., edge labels) found in $p$
(line 5).
For example,
the first path
(Figure~\ref{fig:first-path})
contains three substitutions
(e.g., $[K \mapsto {\tt String}, V \mapsto {\tt String}, E \mapsto T]$),
while the second path
(Figure~\ref{fig:second-path})
includes the following substitutions:
$[K \mapsto X, V \mapsto Y, E \mapsto K]$.
Next,
the algorithm performs constant propagation
on the collected substitutions and creates
a new substitution $\sigma'$ (line 6).
For example,
the substitution of path (1) becomes
$\sigma'=[K \mapsto {\tt String}, V \mapsto {\tt String}, E \mapsto {\tt String}]$,
while the substitution
when handling path (2) is
$[K \mapsto X, V \mapsto Y, E \mapsto X]$.
Using the new substitution $\sigma'$,
we create a new type $t_2$
that is given by $\sigma't'$.
For example,
for the path of Figure~\ref{fig:first-path}
$t_2$ is {\tt Set<String>},
while $t_2$ stands for {\tt Set<X>}
when dealing with the path of Figure~\ref{fig:second-path}.
Then,
the algorithm unifies the type $t_2$
with the target type $t$ (line 8).
Essentially,
this type unification operation identifies
assignments for the free type variables of path $p$
so that the resulting type of the path expression
is compatible with the target type $t$.
If the two types are not unifiable,
the path is considered infeasible,
and the algorithm proceeds to the next iteration
(line 9).
For example,
the type $t_2 = \texttt{Set<String>}$ of Figure~\ref{fig:first-path}
is not unifiable with the target type {\tt Set<Int>}.
This means that
path (1) is an infeasible path,
not leading to an expression of type {\tt Set<Int>}.
When the types are unifiable,
the outcome of type unification ($\sigma_1$---line 8)
instantiates some of the free type variables of $p$.
For example,
consider again Figure~\ref{fig:second-path}:
the algorithm unifies types {\tt Set<Int>}
and {\tt Set<X>}
by returning the substitution $\sigma_1 = [X \mapsto {\tt Int}]$.

In the final step,
the algorithm instantiates
any free type variables of $p$,
not instantiated by 
the previous unification operation (line 10).
Such free type variables
do~\emph{not} affect the type of the underlying expression
and therefore,
the algorithm is free to instantiate them
with any valid type
that respects their upper bounds.
For example,
$Y$ is another free type variable
that stems from the path of Figure~\ref{fig:second-path}.
The algorithm instantiates it with a randomly selected type,
e.g., {\tt String}.
The final substitution $\sigma'$ is then the union
of substitutions $\sigma_1$ and $\sigma_2$ (line 11).
The final substitution and the corresponding path $p$
are added to the set $P$ (line 12).
For example,
the path of Figure~\ref{fig:second-path}
forms an inhabitant of type {\tt Set<Int>}
if using the substitution
$[X \mapsto {\tt Int}, Y \mapsto {\tt String}, K \mapsto {\tt Int},
V \mapsto {\tt String}, E \mapsto {\tt Int}]$.

\point{Termination}
The function \textit{Paths(G, t)}
on line 4 of Algorithm~\ref{alg:find-paths}
enumerates all paths in graph $G$
that reach the target node $t$.
Since an~\graph~may contain cycles,
there might be an infinite number of paths.
To tackle this,
we consider only acyclic paths,
which are guaranteed to be finite.

\setlength{\textfloatsep}{1pt}
\begin{algorithm}[t]
\small
\SetKwProg{Pn}{fun}{=}{}
\SetKwFunction{INHA}{inhabitants}
\SetKwFunction{TOEXPR}{to\_expr}
\DontPrintSemicolon
\SetKwInOut{Input}{Input}
\SetKwInOut{Output}{Output}
\SetKwFor{Case}{case}{$\Rightarrow$}{}%
\SetKwFor{Switch}{match}{with}{}%
\SetInd{0.1em}{1em}
\Pn{\INHA{$G, t$}}{
    $\textit{inhabitants} \gets \emptyset$\;
    \For{$\langle p, \sigma \rangle \in {\tt find\_API\_paths}(G, t)$}{
        $\textit{inhabitants} \gets \textit{inhabitants} \cup \textit{2expr}(p, \sigma)$\;
    }
    \Return{$\textit{inhabitants}$}
}
\caption{Algorithm for converting paths into concrete expressions}
\label{alg:to-expr}
\end{algorithm}

\point{Converting a path into an expression}
Having computing the set of paths
that form inhabitants of a type $t$,
we employ Algorithm~\ref{alg:to-expr}
to convert every path into a concrete expression.
The algorithm employs the function \textit{2expr},
which given an~\graph~$G$,
and a substitution $\sigma$,
converts a path $p$ into 
an expression as follows:

\begin{align*}
    \textit{2expr}(G, p\boldsymbol{\cdot} \prim{fun } m\prim{<}\overline{\alpha}\prim{>}(\overline{x: p}): t, \sigma) &= \textit{2expr}(G, p, \sigma).m\prim{<}\overline{t}\prim{>}(\overline{e})  & \text{where}\\
    &\begin{aligned}[t]
    &\overline{e} = e_i \in {\tt inhabitants}(G, \sigma p_i), \forall i \in\{1\dots n\}\\
    &\overline{t} = (\sigma(\alpha_i))_{i=1}^n \end{aligned}\\
    \textit{2expr}(G, p\boldsymbol{\cdot} \prim{var } f: t, \sigma) &= \textit{2expr}(G, p, \sigma).f\\
    \textit{2expr}(G, d, \sigma) &= \textit{constant}(d) &\text{if } d \in \textit{Type} \\
    \textit{2expr}(G, p\boldsymbol{\cdot} d, \sigma) &= \textit{2expr}(G, p, \sigma) &\text{if } d \in \textit{Type} \\
    \textit{2expr}(G, [], \sigma) &= \epsilon
\end{align*}

For a path consisting
of a single type node,
the function $\emph{2expr}$ simply returns
a constant expression $\textit{constant}(t)$,
which is typically translated into a cast null expression
(e.g., {\tt ???.asInstanceOf[T]} in Scala).
Type nodes found in intermediate positions
within a path are ignored.
To illustrate this path conversion process,
consider the path of Figure~\ref{fig:second-path}.
When using the substitution
$[X \mapsto {\tt Int}, Y \mapsto {\tt String}]$,
this path results in the expression
{\tt Utils.mapOf<Int, String>().keySet()}.

\point{Realization of $\gamma$}
Based on the aforementioned definitions and algorithms,
we finally present the realization
of the concretization function $\gamma$,
previously shown in Definition~\ref{def:concrete}.
The function $\gamma$ is an abstract concept
that generates~\emph{all}~\ir\ expressions that
originate from a given typing sequence and an~\graph.
Our realization function called $\hat{\gamma}$
under-approximates the behavior of $\gamma$,
meaning that
for a given~\graph\ $G$,
a typing sequence $s_d$,
an API component $d$,
and a type substitution,
we have 
$\hat{\gamma}(G, s_d, d, \sigma) \subseteq \gamma(G, s_d, d, \sigma)$.
To do so,
the $\hat{\gamma}$ function
relies on Algorithm~\ref{alg:to-expr}
as follows:
\begin{definition}
\label{def:abstract-concrete}
Let $\hat{\gamma}: G, \textit{Def}, \textit{Type} \times \dots \times \textit{Type} \times \Sigma \longrightarrow \mathcal{P}(\textit{Expr})$ such that:
\begin{align*}
    \hat{\gamma}(G, \bot, t, \sigma) &= \{e \ |\ e \in \texttt{inhabitants}(G, t)\} \\
    \hat{\gamma}(G, \prim{var}\ f: t, \langle t, \bot\rangle, \sigma) &=
    \{e.f\ |\ e \in \texttt{inhabitants}(G, t)\} \\
    \hat{\gamma}(G, \prim{fun}\ m\prim{<}\overline{\alpha}\prim{>}(\overline{x: p}): t, \langle r, \overline{p'}\rangle, \sigma) &=
 \left\{ e_1.m\prim{<}\overline{t}\prim{>}(\overline{e_2})\ \middle| \ 
        \begin{aligned}
            & e_1 \in {\tt inhabitants}(G, r)\\
            & \overline{e_2} = e_{2_i} \in {\tt inhabitants}(G, p_i'), \\ &\forall i \in \{1\dots n\}, \overline{t} = (\sigma(\alpha_i))_{i = 1}^{n}
        \end{aligned}
\right\}\\
    \hat{\gamma}(G, d, s \boldsymbol{\cdot} t, \sigma) &= \{\prim{local var}\ x: t\ |\ e \in \hat{\gamma}(G, d, s, \sigma)\}
\end{align*}
\end{definition}

\begin{theorem}[Soundness]
Given an~\graph~$G$,
an API definition $d \in \textit{Def}$,
a typing sequence $s$,
and a substitution $\sigma$,
if $\hat{\gamma}(G, d, s, \sigma)$ returns $E$,
then $E \subseteq \gamma(G, d, s, \sigma)$.
\end{theorem}
\begin{proof}
The theorem follows straightforwardly
from lines 8--9 of Algorithm~\ref{alg:find-paths}
as we consider only paths that result in a type
unifiable with the target type $t$.
\end{proof}

\begin{theorem}[Completeness]
Given an~\graph~$G$,
an API definition $d \in \textit{Def}$,
a typing sequence $s$,
and a substitution $\sigma$,
$\hat{\gamma}(G, d, s, \sigma) \neq \emptyset$.
\end{theorem}
\begin{proof}
This can be proven by Algorithm~\ref{alg:find-paths}.
The algorithm always returns a non-empty solution,
even if there is no path to the target type $t$.
In this case,
the algorithm returns a singleton path containing
the given type $t$ (line 3, Algorithm~\ref{alg:find-paths}),
which is in turn translated into a constant expression
according to function \textit{2expr}.
\end{proof}

\subsection{Type Erasure}
\label{sec:type-erasure}

Next,
an~\ir~expression obtained from an API typing sequence
(Section~\ref{sec:concretization})
is passed as input to the type erasure process.
The objective is to remove
types from an ~\ir~expression,
while still maintaining the type correctness of the input expression.
Type erasure helps
our approach
stress-test the implementation of
compiler type inference procedures~\cite{hephaestus}.

Type erasure aims to
construct polymorphic invocations
with no explicit type arguments.
Our method relies on
local type inference algorithms~\cite{local},
which are commonly supported
by mainstream programming languages.
In the context of local type inference,
the type arguments of polymorphic invocations
are deduced in a manner that satisfies two conditions:
(1) the method arguments must be compatible
with the formal parameter types of the method being invoked,
and (2) the type of the entire polymorphic invocation
must be compatible with a target type
derived from the surrounding context.

Specifically,
our type erasure process
leverages the following key insight:
a type argument of a polymorphic invocation
can be safely erased
if it can be inferred by
either the arguments of the invocation,
or the target type of the invocation.
At the same time,
the explicit type argument should be the same
as its inferred counterpart.
More formally:
\begin{definition}
\label{def:safe-erasure}
Consider an~\graph~$G$,
a polymorphic method $m$ with signature
$\textit{sig}(G, m) = \langle r, \overline{p}, t\rangle$
and its invocation $e = r.m\prim{<}\overline{t}\prim{>}(\overline{e'})$,
a target type $k$ for $e$,
and a substitution $\sigma$
that maps every type variable of method $m$
to their type arguments in $e$.
We say that the type argument of
a type variable of method $m$
can be~\emph{safely removed} from expression $e$,
when:
\begin{itemize}
\item $\sigma(\alpha) = \sigma'(\alpha)$, where $\sigma' = \textit{unify}(k, t)$
\item or $\sigma(\alpha) = \sigma'(\alpha)$, where $\sigma' = \bigcup_{i=1}^n \textit{unify}(p_i', p_i)$ and $G \vdash \overline{e'}: \overline{p'}$
\end{itemize}
\end{definition}

\begin{figure}[t]
\centering
\small
\begin{mathpar}
\inferrule[constant]{
}
{
    \textit{erasure}(G, \textit{constant(t)}, t') = \textit{constant(t)}
}
\hva \and
\inferrule[field access]{
}
{
    \textit{erasure}(G, e.f, t) = \textit{erasure}(G, e, \bot).f
}
\hva \and
\inferrule[method call]{
    e = e_1.m\prim{<}\overline{t}\prim{>}(\overline{e_2}) \\
    G \vdash \overline{e_2}: \overline{p} \\\\
    \forall \alpha \in \textit{TypeVar}(m). \textit{the type argument of }\alpha\ \textit{can be safely erased from e when the target type is k}
}
{
    \textit{erasure}(G, e, k) = \textit{erasure}(G, e_1, \bot).m(\textit{erasure}(G, \overline{e_2}, \overline{p})\dots)
}
\hva \and
\inferrule[local var]{
}
{
    \textit{erasure}(G, \primsem{local var } x: t = e, t') = \primsem{local var } x: t = \textit{erasure}(G, e, t)
}
\end{mathpar}
\vspace{-4mm}
\caption{Defintion of type erasure. The~\textit{erasure}
function ($G \times \textit{Expr} \times \textit{Type} \longrightarrow \textit{Expr}$)
takes an~\graph,
an expression $e$,
and a target type $t$,
and yields another expression $e'$
with erased type information.}
\label{fig:type-erasure}
\vspace{2mm}
\end{figure}

Based on the preceding property,
we define the function $\textit{erasure}$
($G \times \textit{Expr} \times \textit{Type} \longrightarrow \textit{Expr}$),
which takes an~\graph,
an input expression $e$,
and a target type $t$.
It outputs another expression
with type information removed.
The full definition of $\textit{erasure}$
is shown in Figure~\ref{fig:type-erasure}.
The function $\textit{erasure}$ is recursively applied
to any sub-expression included in $e$.
When handling a method call $e$,
the type arguments are removed by $\textit{erasure}$,
but only when~\emph{every} type argument
of the call can be safely erased
according to Definition~\ref{def:safe-erasure}
([{\sc method call}]).

\point{Example}
Consider the following code snippet:

\begin{minipage}{0.5\linewidth}
\begin{lstlisting}[language=Java]
<T> void m1(T x) {}
<X, Y> Y m2(X p1)
m1<Object>("str");
m1<String>("str");
m2<String, String>("f");
String x = m2<String, String>("str");
\end{lstlisting}
\end{minipage}
~
\begin{minipage}{0.5\linewidth}
\begin{lstlisting}[language=Java, numbers=none]
<T> void m1(T x) {}
<X, Y> Y m2(X p1)
m1<Object>("str"); // Not erased
m1("str"); // Erased
m2<String, String>("f"); // Not erased
String x = m2("str"); // Erased
\end{lstlisting}
\end{minipage}
After applying the (optional) function $\textit{erasure}$,
we get the method calls shown on the right.
\textit{erasure}
removes the type arguments from the second and
fourth polymorphic invocation (lines 4, 6),
while the first and the third call remain the same.
Specifically,
if we chose
to erase the type argument of the first call (line 3),
the inferred type of {\tt T} would become {\tt String},
which is not equivalent with the explicit type
argument {\tt Object}.
Similarly,
we do not remove the type arguments from the method call
on line 5,
because the type argument of type variable {\tt Y}
cannot be safely erased.
This is because
there is no target type
for the entire call
that helps with the inference of {\tt Y}.

\point{Relation to~\heph' type erasure approach}
Erasing types has proven useful in finding
bugs in type inference implementations,
as demonstrated by the~\heph\ tool~\cite{hephaestus}.
\heph\ comes with a mutation
that erases types from an~\emph{existing} program.
In contrast,
our approach differs
by incorporating type erasure
directly into the synthesis process.
During the generation of a method call,
we check whether the method call's type arguments
can be omitted based on the inference
rules outlined in Figure~\ref{fig:type-erasure}.
Notably,
integrating type erasure into
the synthesis process makes our approach efficient,
as we further
show in Section~\ref{sec:eval-thalia-hephaestus} of our evaluation.
One distinct difference is the overhead
in~\heph,
which is attributed to an intra-procedural analysis
that preserves the type correctness of the input program
during the type erasure mutation.

\subsection{Implementation and Discussion}
\label{sec:implementation}

Having presented the theory behind
the main components of the approach,
this section focuses on noteworthy
technical details.
We have implemented our techniques
on top of~\heph,
the modern framework
for testing compilers' type checkers~\cite{hephaestus}.
Our implementation,
which we call~\tool,
extends~\heph~using roughly~\nnum{5k} additional
lines of Python code.

The input of~\tool~is a set of JSON files
that describe a given API.
We have developed an auxiliary script
that automatically produces such JSON files
by parsing a library's API documentation web pages
(e.g., generated by {\tt javadoc}) via
the {\tt beatifulsoup4} Python package.
\tool~then examines the input JSON files
and builds the corresponding~\graph.

\point{Producing typing sequences}
API enumeration is exponential
in terms of the number of types
found in the signature of an API component.
To make API enumeration practical,
\tool~employs randomization.
When dealing with a polymorphic API component $d$,
\tool~first generates a~\emph{valid} type substitution $\sigma$
at random.
Based on the randomly generated substitution $\sigma$,
\tool~then computes the set $S_d$ containing the typing sequences of
$d$ as described in Section~\ref{sec:api-enumeration}.
\tool\ generates one test case
for every $s_d \in S_d$
by applying function $\hat{\gamma}(G, d, s_d, \sigma)$
as detailed in Section~\ref{sec:concretization}.

\point{Enumerating API paths}
To compute all simple paths
that reach a specific node,
\tool~employs Yen's algorithm~\cite{yens},
which computes the $k$-shortest loopless paths in a graph.
During the concretization of an API typing sequence,
\tool~invokes a variant of Algorithm~\ref{alg:find-paths}
presented in Section~\ref{sec:concretization}.
\tool~iterates over the paths given by Yen's algorithm
in a~\emph{random} order,
and rather than returning all feasible paths,
\tool~returns the first random path that forms a type inhabitant.


\point{\graph~size}
Intuitively,
the input API affects the size of the underlying~\graph,
and thus the scalability of our approach.
However,
as we show in our evaluation,
our tool can easily handle real-world APIs
with more than~$\empirical{25}$k edges,
and synthesize programs in milliseconds.

\point{Incomplete APIs}
Although our realization function $\hat{\gamma}$
(Section~\ref{sec:concretization})
always returns a non-empty set of expressions,
we may be unable to exercise a specific API component
due to the presence of recursive bounds and
the use of an API with missing type information.
To illustrate this,
consider:
\begin{lstlisting}[numbers=none,language=Java]
class A<T extends A<T>> {  int getSize();  }
\end{lstlisting}
We are unable to produce a typing sequence
that invokes {\tt getSize}
as we cannot construct a proper receiver type 
derived from the type constructor {\tt A}.
This is because,
there is no valid instantiation of type variable {\tt T}
that is compatible with the bound {\tt A<T>}.
This is unavoidable
because the preceding API does not contain subtypes of {\tt A}
(e.g., a subclass {\tt class B extends A<B>}).
To tackle this,
our enumeration skips all API entities for
which we cannot produce well-typed typing sequences.

\point{Generalizability}
\label{general}
\tool~currently produces programs written
in three popular languages: Scala, Kotlin, and Groovy.
Adding a new language requires
(1) the collection
of its APIs,
(2) the implementation of
a parser that transforms
the string representation of a type
(as it appears in the input JSON)
into its in-memory counterpart supported by~\ir,
and (3) a translator
to convert~\ir~programs into source files
written in the target language.

Although
our approach enables increased feature coverage
(see Section~\ref{sec:eval-test-case-chars}),
we might encounter an API
that exhibits type system-related features
that are not currently understood by~\ir.
By default,
we skip exercising API entities
that involve such unsupported features.
As a result,
our current implementation
might not be as effective for languages,
such as Rust, OCaml,
or TypeScript
because the type system of API-IR
does not currently support many of their core type system features,
including currying, type aliases, structural types,
associated types or more expressive bound constraints.
However,
the fundamental concepts of our approach
(e.g., API enumeration) are still applicable to any modern
language,
because APIs are ubiquitous.
For example,
one has the option to enhance our~\ir~language
with new type-related features.
We have already done so
to accommodate Scala's higher-kinded types,
and Kotlin's nullable types.

\section{Evaluation}
\label{sec:evaluation}

Our evaluation answers the following research questions:
\begin{enumerate}[label={\bf RQ\arabic*}, leftmargin=2.3\parindent]
    \item Is~\tool~effective in finding new compiler typing bugs?
        (Section~\ref{sec:eval-bug-finding})
        \point{Short answer}
        During a five-month period 
        of developing~\tool~and testing industrial-strength compilers, 
        a total of~\ttotal~bugs were detected by~\tool. 
        Out of these,~\treal~bugs have been confirmed or fixed.
    \item What are the characteristics of the test cases generated by~\tool?
        (Section~\ref{sec:eval-test-case-chars})
        \point{Short answer}
        \tool's test cases are concise, 
        averaging only 11--13 lines of code 
        across various languages. 
        To reveal bugs,
        these compact test cases
        effectively combine key features,
        such as parametric polymorphism,
        overloading, and higher-order functions.
        Rather than solely relying on a generative process,
        these features originate from~\emph{existing} API definitions.
    \item What is the impact of libraries
        and synthesis modes
        on the effectiveness of~\tool? 
        (Section~\ref{sec:eval-rq3})
        \point{Short answer}
        Using diverse libraries as input seeds 
        increases code coverage 
        across all tested compilers. 
        Different synthesis settings further boost
        (1) the bug-finding capability
        by uncovering~\empirical{36} additional bugs,
        and (2) the line coverage
        by up to $\empirical{5}$\%--$\empirical{11.9}$\%
        across all compilers.
        \tool~scales well
        even with libraries containing large APIs, 
        with an average synthesis time (per program)
        of $\empirical{161}$ to $\empirical{256}$ milliseconds,
        depending on the target language.
    \item How does~\tool~compare to the state-of-the-art,
        namely~\heph?
        (Section~\ref{sec:eval-thalia-hephaestus})
        \point{Short answer}
        \tool~and~\heph~complement each other effectively.
        \tool~has detected~at least~\empirical{42} bugs missed by prior work.
        \heph~and~\tool~achieve similar code coverage.
        However, 
        combining the results of both tools 
        significantly increases code coverage
        (up to~\empirical{4.4}\%),
        demonstrating \tool's ability to test
        previously unexplored components.
\end{enumerate}

\subsection{Experimental Setup}
\label{sec:eval-setup}

\point{Hardware and compiler versions}
We performed all experiments on commodity servers (32
cores and 64 GB of RAM per machine)
running Ubuntu 22.04 (x86\_64).
Our testing efforts focused on
the compilers of Groovy, Kotlin, and Scala.
Although~\tool~can produce
Java programs,
we did not consider Java in our evaluation.
The reason is that
the issue tracker of OpenJDK
is not open to the public,
and therefore we were unable to directly interact
with its compiler's developers.
For each compiler,
we tested (1) its latest development version
by regularly building the compiler on
its latest commit,
or (2) its most recent stable version.

\point{API collection}
We collected APIs from two sources.
Initially,
we examined the APIs
found in the standard library
of the corresponding language,
such as collections API,
I/O API.
Beyond standard libraries,
we also considered APIs from third-party libraries
found in the Maven central repository.
Specifically,
we examined the Maven repository
and Scala index~\cite{scaladex}
to gather the group ID,
artifact ID,
and the latest versions
of the most popular Kotlin, Java,
and Scala libraries.
Rather than selecting libraries solely based on their popularity,
we could establish a feature coverage criterion.
This criterion would (1) exclude libraries that exhibit
the same features with previously explored libraries,
or (2) prioritize libraries with new and complex typing features.
This is not a straightforward task,
so we leave it as future work.

For each library,
we searched in the Maven central repository
to fetch:
(1) the API documentation of the library,
(2) the JAR files of the library and
its dependencies.
The documentation was converted into
JSON files given as input to~\tool.
Then,
we built a classpath that pointed to the retrieved JAR files
so that the compiler under test could
locate the external symbols
defined in external libraries.
In total,
we selected~\empirical{95} libraries
whose characteristics are shown in 
Table~\ref{tab:lib-chars}.
The entry ``other'' shows the average metrics
of the libraries not shown in the table.

We ran~\tool~on each selected API
under four different synthesis modes.
In particular,
on each API,
we used~\tool~to synthesize
(1) well-typed programs
with and without type erasure,
and (2) wrongly-typed programs
with and without type erasure.

\begin{table}[t]
\caption{Characteristics of ten selected libraries.
Each table entry indicates the number of nodes and edges
of the~\graph,
the number of methods (including the polymorphic ones),
the number of fields,
the number of constructors,
the number of types (including type constructors),
the average size of the inheritance chain,
and the average size of API signatures.}
\label{tab:lib-chars}
\centering
\resizebox{\linewidth}{!}{\begin{tabular}{lrrrrrrrr}
\hline
{\bf Library} & {\bf Nodes} & {\bf Edges}  & {\bf Poly. M/Methods}    & {\bf Fields} & {\bf Constr} & {\bf Type Con/Types} & {\bf Inhr size} & {\bf Sig size} \bigstrut\\
\hline
\libaname & \libanodes & \libaedges & \libapolym /\libamethods & \libafields & \libacons & \libatypecon /\libatypes & \libainh & \libasig \\
\libbname & \libbnodes & \libbedges & \libbpolym /\libbmethods & \libbfields & \libbcons & \libbtypecon /\libbtypes & \libbinh & \libbsig \\
\libcname & \libcnodes & \libcedges & \libcpolym /\libcmethods & \libcfields & \libccons & \libctypecon /\libctypes & \libcinh & \libcsig \\
\libdname & \libdnodes & \libdedges & \libdpolym /\libdmethods & \libdfields & \libdcons & \libdtypecon /\libdtypes & \libdinh & \libdsig \\
\libename & \libenodes & \libeedges & \libepolym /\libemethods & \libefields & \libecons & \libetypecon /\libetypes & \libeinh & \libesig \\
\libfname & \libfnodes & \libfedges & \libfpolym /\libfmethods & \libffields & \libfcons & \libftypecon /\libftypes & \libfinh & \libfsig \\
\libgname & \libgnodes & \libgedges & \libgpolym /\libgmethods & \libgfields & \libgcons & \libgtypecon /\libgtypes & \libginh & \libgsig \\
\libhname & \libhnodes & \libhedges & \libhpolym /\libhmethods & \libhfields & \libhcons & \libhtypecon /\libhtypes & \libhinh & \libhsig \\
\libiname & \libinodes & \libiedges & \libipolym /\libimethods & \libifields & \libicons & \libitypecon /\libitypes & \libiinh & \libisig \\
\libjname & \libjnodes & \libjedges & \libjpolym /\libjmethods & \libjfields & \libjcons & \libjtypecon /\libjtypes & \libjinh & \libjsig \\
\libkname & \libknodes & \libkedges & \libkpolym /\libkmethods & \libkfields & \libkcons & \libktypecon /\libktypes & \libkinh & \libksig \\
\hline
{\bf \avgname} & {\bf \avgnodes} & {\bf \avgedges} & {\bf \avgpolym /\avgmethods} & {\bf \avgfields} & {\bf \avgcons} & {\bf \avgtypecon /\avgtypes} & {\bf \avginh} & {\bf \avgsig} \bigstrut \\
\hline
\end{tabular}}
\end{table}

\point{Configuration}
\tool~comes with a set of constants
that affect the number of the synthesized programs
(Section~\ref{sec:implementation}).
For enumerating ill-typed typing sequences,
we configured~\tool~so that
every type set contains at most five
incompatible types selected randomly.
Regarding type inhabitant selection,
we configured Yen's algorithm
to return the shortest path
between two nodes.
This is the equivalent to calling
Dijkstra's algorithm for the shortest path calculation.
Finally,
the maximum depth of the generated programs is two.
Based on our exploratory experiments,
these configurations
provide a good balance between bug-finding capability
and performance.

\subsection{RQ1: Bug-Finding Results}
\label{sec:eval-bug-finding}

\begin{table*}[t]
\caption{
(a) Status of the reported bugs in {\tt groovyc}, {\tt kotlinc}, and Dotty,
(b) number of bugs with
unexpected compile-time error (UCTE),
unexpected runtime behavior (URB),
crash,
or compilation performance (CPI) symptom,
(c) bugs revealed by
well-typed programs,
well-typed programs with erased types
(TE),
or ill-typed ones.
}
\label{tab:bugs}
\begin{subtable}[T]{.32\linewidth}
\centering
\resizebox{\linewidth}{!}{\begin{tabular}{lrrr|r}
\hline
{\bf Status} & {\bf groovyc} & {\bf kotlinc}  & {\bf Dotty}    & {\bf Total} \bigstrut\\
\hline
Confirmed    & \gconfirmed     & \kconfirmed   & \sconfirmed   & \tconfirmed   \\
Fixed        & \gfix           & \kfix         & \sfix         & \tfix         \\
Duplicate    & \gdupl          & \kdupl        & \sdupl        & \tdupl        \\
Won't fix    & \gwont          & \kwont        & \swont        & \twont        \\
\hline
Total        & {\bf \gtotal}   & {\bf \ktotal} & {\bf \stotal} & {\bf \ttotal} \bigstrut\\
\hline
\end{tabular}}
\subcaption{\label{tab:reported}}
\end{subtable}  
\hfill
\begin{subtable}[T]{.32\linewidth}
\centering
\resizebox{\linewidth}{!}{\begin{tabular}{lrrr|r}
\hline
{\bf Symptom} & {\bf groovyc} & {\bf kotlinc}  & {\bf Dotty}    & {\bf Total} \bigstrut\\
\hline
UCTE          & \gucte       & \kucte        & \sucte        & \tucte        \\
URB           & \gurb        & \kurb         & \surb         & \turb         \\
Crash         & \gcrash      & \kcrash       & \scrash       & \tcrash       \\
CPI           & \gcpi        & \kcpi         & \scpi         & \tcpi         \\
\hline
\end{tabular}}
\subcaption{\label{tab:symtpoms}}
\end{subtable}  
\hfill
\begin{subtable}[T]{.32\linewidth}
\centering
\resizebox{\linewidth}{!}{\begin{tabular}{lrrr|r}
\hline
{\bf Program type} & {\bf groovyc} & {\bf kotlinc}  & {\bf Dotty}    & {\bf Total} \bigstrut \\
\hline
Well-typed    & \ggenerator  & \kgenerator   & \sgenerator   & \tgenerator   \\
Well-typed (TE)  & \ginference  & \kinference   & \sinference   & \tinference   \\
Ill-typed     & \gsoundness  & \ksoundness   & \ssoundness   & \tsoundness   \\
\hline
\end{tabular}}
\subcaption{\label{tab:components}}
\end{subtable}  
\vspace{-2.5mm}
\end{table*}

Table~\ref{tab:reported} summarizes
the bug-finding results of~\tool.
Overall,
within five months of concurrent
development and testing,
we reported~\ttotal bugs to developers,
\treal~of which are either
confirmed or fixed.
To prevent duplicate bug reports, 
we conducted a comprehensive search 
in the issue tracker 
prior to submission, 
to identify potential related bugs.
After reporting a small number of bugs
to each development team,
we waited a couple of weeks
for the developers to fix or triage them.
This helped us avoid
overwhelming developers
with an extreme number of unfixed/untriaged bugs.
In the end,
\empirical{four} of the reported bugs
were marked as ``won't fix''.
One ``won't fix'' issue
is a compiler crash in Dotty.
Scala developers think that it is not worth
fixing it,
although the issue is a regression
introduced in Scala 3.
Another one is a known limitation issue
of {\tt groovyc}.

\point{Symptoms of compiler typing bugs}
The symptoms of the discovered bugs
are shown in Table~\ref{tab:symtpoms}.
The majority
(\tucte/\ttotal)
lies in unexpected compile-time errors (UTCE):
a compiler~\emph{unintentionally} rejects
a well-typed program.
UTCE is followed by compiler crash
(\tcrash)
and unexpected runtime behavior
(\turb).
An unexpected runtime behavior (URB)\footnote{
We follow the same terminology
used in the study of~\citet{typing-study}.
\tool\ detects URB bugs at compile-time,
and not at runtime,
based on our test oracle for soundness bugs,
that is,
erroneous acceptance of ill-typed programs.}
is the symptom of a compiler bug
that becomes evident at runtime.
For example,
soundness bugs typically lead to URB symptoms,
e.g., runtime errors that should have been caught at compile-time.
Two discovered bugs
trigger severe compilation performance issues (CPI)
in the type checkers of {\tt groovyc} and {\tt kotlinc}.

\point{Bug-finding under different synthesis modes}
Table~\ref{tab:components} shows
the bug-finding results of
each synthesis mode.
Most of the bugs
(\tgenerator/\ttotal)
are triggered in the base mode,
where~\tool~produces
well-typed programs
that contain full type information
(i.e., type erasure is disabled).
This suggests that
\tool~is effective in detecting
compiler typing bugs,
even when generating simple client code 
(at least at first blush).
Using the type erasure process,
\tool~has identified~\tinference additional bugs.
All of these bugs are related to issues
found in the type inference procedures
of the compilers.
Finally,
\tool~has discovered~\tsoundness bugs that
originate from ill-typed code.
The low number of bugs triggered by ill-typed code
in relation to the number of bugs that arise from well-typed programs
is consistent with the findings
of~\citet{typing-study,hephaestus}.

\point{Intricacy of typing bugs}
As highlighted by the study of~\citet{typing-study},
fixing compiler typing bugs
often becomes particularly challenging.
This insight is indeed
confirmed by most of the compiler development teams
we have interacted with:
many reported bugs require much expertise
in the areas of overload resolution and type inference.

Our interaction with developers has revealed
another noteworthy observation:
typing bugs take a fair amount of effort 
to understand their root causes,
particularly
identifying which specific components
in type checking procedure break down.
Therefore,
additional tooling for root cause analysis would be
beneficial for compiler developers.
Furthermore,
\tool~has helped us reveal interesting bugs 
that stem from issues
in the language design.
These bugs are exceptionally challenging
to fix,
as their fix requires pervasive changes
in the compiler code base.
These changes are susceptible to
breaking other language features.

\subsection{RQ2: Test Case Characteristics}
\label{sec:eval-test-case-chars}

\begin{table*}[t]
\caption{
The language features that appear in
the minimized test cases of our reported bugs.
A check mark indicates whether
the corresponding feature is supported by
the state-of-the-art tool~\heph.
}
\label{tab:features}
\begin{subtable}[T]{.56\linewidth}
\centering
\resizebox{\linewidth}{!}{\begin{tabular}{llrr}
\hline
& & & {\bf Supported by} \\
{\bf Feature type} & {\bf Feature} & {\bf Frequency}  & \heph \bigstrut\\
\hline
\multirow{12}{*}{Declaration}          & Polymorphic class                                          & \parameterizedclass /\ttotal            & \cmark     \\
                     & Polymorphic function                                       & \parameterizedfunction /\ttotal         & \cmark     \\
                     & Single Abstract Method (SAM)                                 & \sam /\ttotal                           & \cmark     \\
                     & Overloading                                                  & \overloading /\ttotal                   & \xmark     \\
                     & Inheritance/Implementation of multiple interfaces          & \inheritance /\ttotal                   & \xmark     \\
                     & Variable argument                                            & \variableargument /\ttotal              & \cmark     \\
                     & Access modifier                                              & \accessmodifier /\ttotal                & \xmark     \\
                     & Inner class                                                  & \innerclass /\ttotal                    & \xmark     \\
                     & Bridge method                                                & \bridgemethod /\ttotal                  & \xmark     \\
                     & Default method                                               & \defaultmethod /\ttotal                 & \xmark     \\
                     & Static method                                                & \staticmethod /\ttotal                  & \xmark     \\
                     & Operator                                                     & \operator /\ttotal                      & \xmark     \\
\hline
\multirow{2}{*}{Type inference} & Type argument inference                                      & \typeargumentinference /\ttotal         & \cmark     \\
                     & Variable type inference                                      & \variabletypeinference /\ttotal         & \cmark     \\
\hline
\end{tabular}}
\end{subtable}  
\hfill
\begin{subtable}[T]{.4\linewidth}
\centering
\resizebox{\linewidth}{!}{\begin{tabular}{llrr}
\hline
& & & {\bf Supported by} \\
{\bf Feature type} & {\bf Feature} & {\bf Frequency}  & \heph \bigstrut\\
\hline
\multirow{8}{*}{Type} & Polymorphic type                                           & \parameterizedtype /\ttotal             & \cmark     \\
                     & Wildcard type                                                & \wildcardtype /\ttotal                  & \cmark     \\
                      & Bounded type parameter                                       & \boundedtypeparameter /\ttotal          & \cmark     \\
                      & Array type                                                   & \arraytype /\ttotal                     & \cmark     \\
                      & Subtyping                                                    & \subtyping /\ttotal                     & \cmark     \\
                      & Recursive upper bound                                        & \recursiveupperbound /\ttotal           & \xmark     \\
                      & Primitive type                                               & \primitivetype /\ttotal                 & \cmark     \\
                      & Nullable type                                                & \nullabletype /\ttotal                  & \xmark     \\
\hline
\multirow{3}{*}{Expression} & Function reference                                           & \functionreference /\ttotal             & \cmark     \\
                      & Lambda                                                       & \llambda /\ttotal                       & \cmark     \\
                      & Conditionals                                                 & \conditionals /\ttotal                  & \cmark     \\
\hline
\end{tabular}}
\end{subtable}  
\vspace{2.5mm}
\end{table*}

We manually identified what language features
appear in every minimized test case
that accompanies our bug reports.
Table~\ref{tab:features} shows the
frequency of the identified features.
Features related to parametric polymorphism
(e.g., polymorphic function, wildcard type, etc.)
are the top features in terms of
bug-finding capability.
Roughly~\empirical{85\%} (\empirical{70}/\ttotal)
of the bug-triggering test cases
contain at least one feature related to parametric polymorphism.
This finding is in line with the work
of~\citet{typing-study,hephaestus}.
Our finding is further supported by compiler developers
who have acknowledged that working
with generics is a highly demanding task
that requires a high level of expertise.
Parametric polymorphism is followed by
functional programming features
(e.g., lambdas, function references),
type inference features,
and overloading.
Figure~\ref{fig:venn} shows the overlap
of parametric polymorphism,
functional programming,
and overloading.
Overall,
these features are found in all but~\empirical{eight} test cases.
Finally,
Table~\ref{tab:features} confirms
that~\tool~increases feature coverage,
especially for definition-related features.
This means that~\tool~can exercise
the implementation of these features in an agnostic way,
as long as they they are part of the input APIs.
\begin{wrapfigure}{r}{0.37\textwidth}
\centering
\includegraphics[scale=0.23]{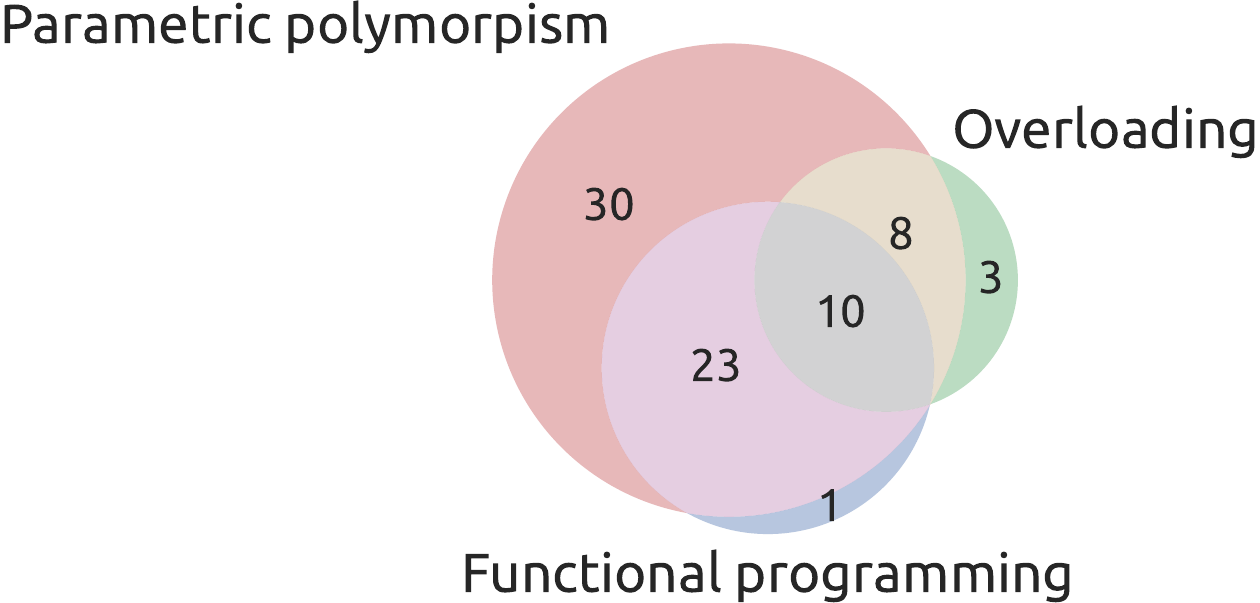}
\vspace{-2mm}
\caption{The Venn diagram of features related to
parametric polymorphism, functional programming,
and overloading.}
\label{fig:venn}
\end{wrapfigure}

\point{Size of test cases}:
The test cases given by~\tool~are
consistently small across the examined languages.
On average,
a Groovy test case measures~$\groovysizeavg$kB
and contains~$\groovylocavg$~lines of code (LoC),
while a Scala test case measures~$\scalasizeavg$kB
and consists of~$\scalalocavg$~LoC.
Similarly,
Kotlin test cases have an average size of~$\kotlinsizeavg$kB
and an average of~$\kotlinlocavg$~LoC. 
The compact size of these test cases
facilitates efficient testing with higher throughput
and simplifies the generation of minimal test-cases.

Overall,
the above results indicate
that although the test cases produced by~\tool~are small in size,
they are able to discover bugs
triggered by a combination of complicated languages features.

\subsection{RQ3: Impact of Library Selection and Synthesis Modes}
\label{sec:eval-rq3}

\begin{figure*}[t]
    \centering
    \begin{subfigure}{0.33\linewidth}
        \centering
        \includegraphics[width=\linewidth]{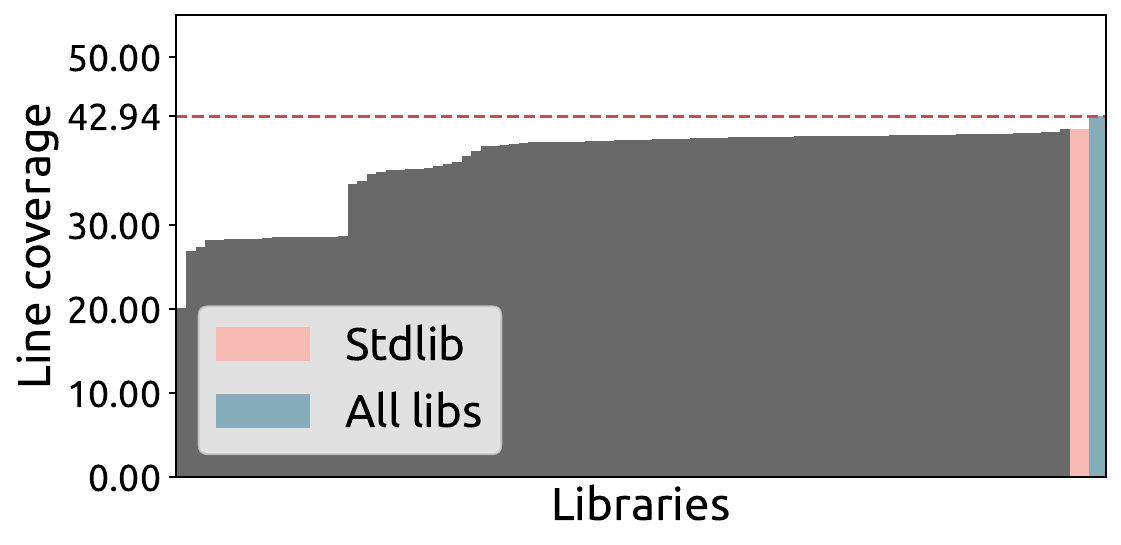}
        \caption{groovyc}
        \label{fig:lib-coverage-groovyc}
    \end{subfigure}%
    \hfill
    \begin{subfigure}{0.33\linewidth}
        \centering
        \includegraphics[width=\linewidth]{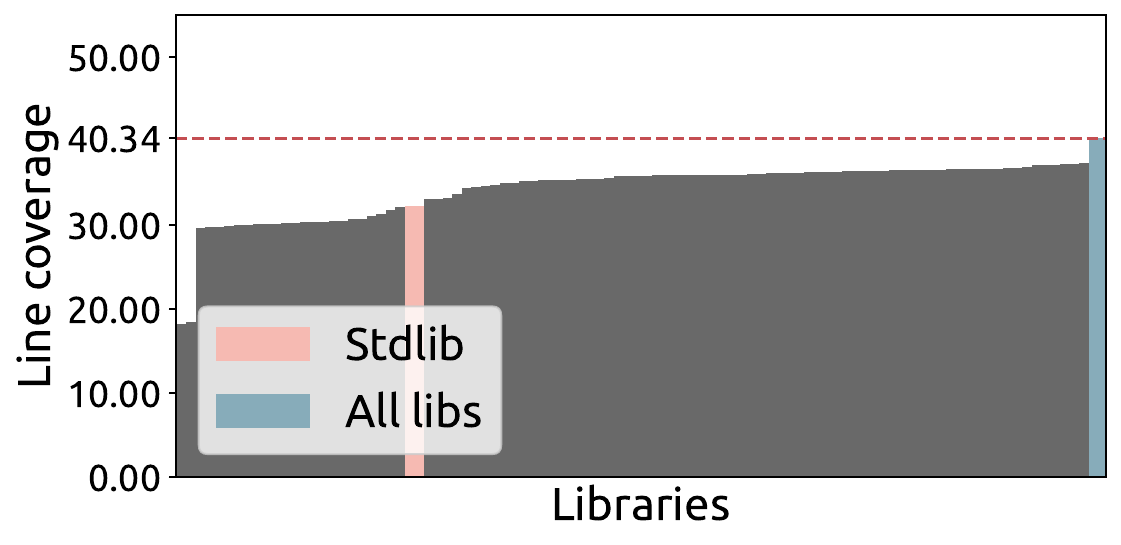}
        \caption{Dotty}
        \label{fig:lib-coverage-dotty}
    \end{subfigure}%
    \hfill
    \begin{subfigure}{0.33\linewidth}
        \centering
        \includegraphics[width=\linewidth]{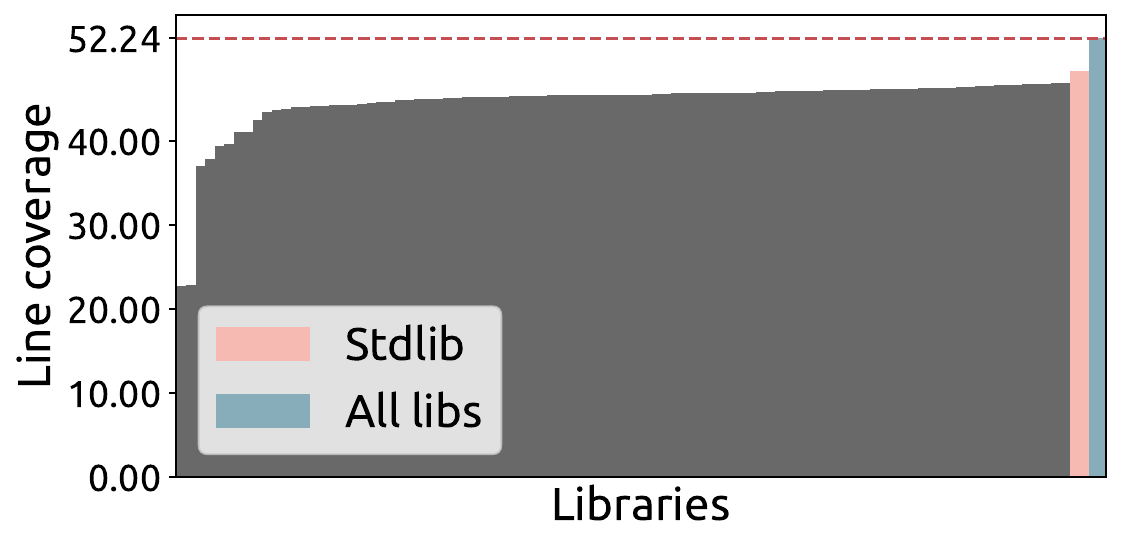}
        \caption{kotlinc}
        \label{fig:lib-coverage-kotlinc}
    \end{subfigure}
    \vspace{-2mm}
    \caption{
        Compiler line coverage when using top 95 Maven libraries.
        Notably,
        when combining the results of all libraries (highlighted in blue),
        there is an increase in line coverage, ranging between
        $\empirical{1.6}$\%--$\empirical{4}$\% when
        compared to the best individual library of each language.
    }
    \label{fig:lib-coverage}
    \vspace{3mm}
\end{figure*}

To evaluate the impact of library selection
and investigate the effectiveness of different synthesis modes, 
we conducted a comprehensive evaluation 
using \empirical{95} top libraries found in the Maven repository, 
including the standard library of each language.

Figure~\ref{fig:test-cases-libs} shows
how many well-typed and ill-typed test cases (typing sequences)
have been produced by our API enumeration
technique.
In total,
the examined~\empirical{95} libraries
yielded roughly~\empirical{1.6} million programs
for each language.
Based on these programs,
we performed an analysis on
the time spent synthesizing them,
as well as the code coverage achieved by
each library and mode.

\point{Code coverage analysis}
We used the JaCoCo library~\cite{jacoco}
to measure the code coverage in each compiler.
Figure~\ref{fig:lib-coverage} shows
the line coverage achieved by
programs derived from each library
in all synthesis modes.
We observe that
the input library and its corresponding API
significantly affects the resulting line coverage.
Libraries with limited methods
and polymorphic components tend
to have lower line coverage,
indicating that the extent of code coverage
depends on the richness of the input API.
When combining the programs that arise
from all the input libraries (see blue bars),
there is a noteworthy code coverage increase
in all the examined compilers.
For example, 
in {\tt groovyc}, 
the combination of libraries
surpasses the line coverage achieved by the best individual library 
by~\empirical{$1.6$\%}. 
This percentage translates to hundreds of additional covered lines.
Similarly, 
in Dotty and {\tt kotlinc}, 
the line coverage improvement ranges between 
$\empirical{3\%}$ and $\empirical{4\%}$. 
Similar patterns have been observed
when considering branch and function coverage. 
Based on the above,
using APIs with distinct characteristics avoids saturation
and exercises unexplored code.

Table~\ref{tab:coverage-base-statistics}
illustrates the impact of each mode 
on line coverage
in comparison to the base mode,
where~\tool~produces well-typed programs
with type erasure disabled.
\tool's different modes
contribute to a line coverage increase up to
$\empirical{5}\%$--$\empirical{11.9}\%$
across all compilers.
While the goal is not to explore the entire code base,
but rather to exercise specific compiler parts,
such as type inference,
employing multiple modes
is crucial in maximizing the effectiveness of~\tool.
This is further supported by
the number of bugs (\empirical{36}/\ttotal)
that were discovered when using a mode
other than the base mode (see Table~\ref{tab:components}).
Many of them concern
important type inference and soundness issues.
\begin{figure}[t]
  \begin{minipage}{0.42\linewidth}
    \centering
    \includegraphics[width=0.8\linewidth]{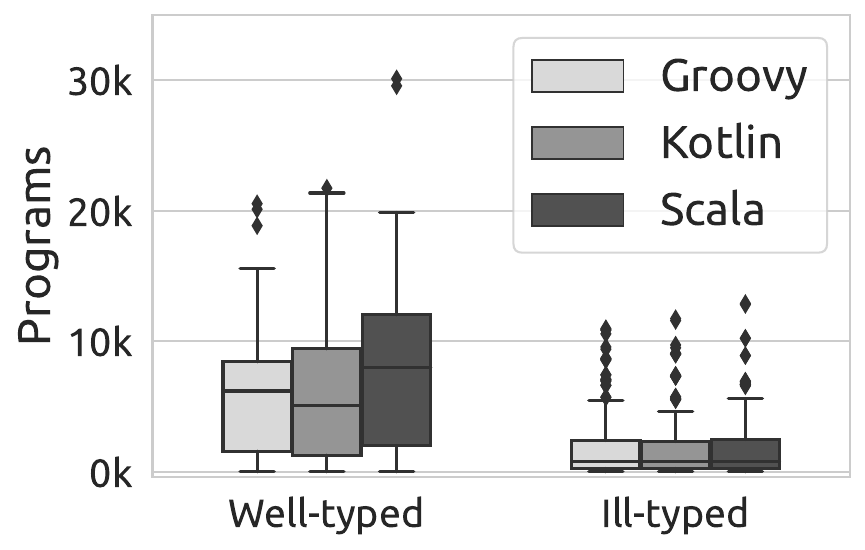}
    \captionof{figure}{Number of test cases produced by \tool~categorized by their type.}
    \label{fig:test-cases-libs}
  \end{minipage} \hfill
  \begin{minipage}{0.57\linewidth}
    \centering
    \captionof{table}{Statistics regarding the increase in line coverage
        achieved by each synthesis mode. The baseline is when \tool\ synthesizes
        well-typed programs with complete type information.
        Numbers outside parentheses indicate absolute change,
        while percentages in parentheses show relative change.
        The type erasure mode (TE) exhibits the highest line coverage increase,
followed by producing ill-typed programs with missing type information (Both).}
    \vspace{-2mm}
    \label{tab:coverage-base-statistics}
    \resizebox{\linewidth}{!}{%
      \begin{tabular}{lclcccccc}
        \hline
        \multirow{2}{*}{Compiler} & \multirow{2}{*}{Mode} & \multicolumn{4}{c}{Statistics} \\ \cline{3-7}
        & & Min & Max & Median & Mean & Std \\ \hline
        \multirow{3}{*}{{\tt groovyc}}
        & TE         & +0 (0\%) & +\nnum{4306} (11.9\%) & +166 (0.5\%) & +265 (0.7\%) & 1.6\% \\
        & Ill-typed  & +0 (0\%) & +\nnum{2952} (8.1\%)  & +45 (0.1\%)  & +118 (0.3\%) & 1.1\% \\
        & Both       & +0 (0\%) & +\nnum{3087} (8.5\%)  & +163 (0.5\%) & +195 (0.5\%) & 0.9\% \\ \cdashline{0-7}
        \multirow{3}{*}{Dotty}
        & TE        & +0 (0\%)       & +\nnum{5548} (7.1\%) & +934 (1.2\%) & +1087 (1.4\%) & 1.1\% \\
        & Ill-typed & +7 (0.01\%)    & +\nnum{3595} (4.6\%) & +475 (0.6\%) & +615 (0.8\%) & 0.7\% \\
        & Both      & +117 (0.2\%)   & +\nnum{4192} (5.3\%) & +852 (1.1\%) & +975 (1.24\%) & 0.8\% \\ \cdashline{0-7}
        \multirow{3}{*}{{\tt kotlinc}}
        & TE        & +3 (0.01\%) & +\nnum{2144} (3.7\%)  & +523 (0.9\%) & +566 (1\%) & 0.5\% \\
        & Ill-typed & +1 (0\%)    & +\nnum{3030} (5.2\%)  & +175 (0.3\%) & +291 (0.5\%) & 0.7\% \\
        & Both      & +45 (0.1\%) & +\nnum{2175} (3.7\%)  & +347 (0.6\%) & +420 (0.7\%) & 0.6\% \\ \hline
      \end{tabular}%
    }
  \end{minipage}
\end{figure}

\point{Synthesis time}
Figure~\ref{fig:time} shows
how the size of the underlying~\graph~influences
the time spent on the synthesis of
a~\emph{single} test case.
The correlation between the size of the~\graph~and
the synthesis time is almost linear
for all the compilers.
\tool~takes less than~\empirical{400} milliseconds
to synthesize one program
coming from large libraries
with more than~\empirical{25k} edges
in the corresponding~\graph.
The average synthesis time
is $\empirical{256}$,
$\empirical{161}$,
and~$\empirical{178}$
milliseconds for Groovy, Scala,
and Kotlin programs respectively.
This demonstrates the practicality of~\tool~
in synthesizing client programs
in milliseconds,
even if they come from large libraries
and complex APIs.

\begin{figure*}[t!]
    \centering
    \begin{subfigure}{0.33\linewidth}
        \centering
        \includegraphics[width=\linewidth]{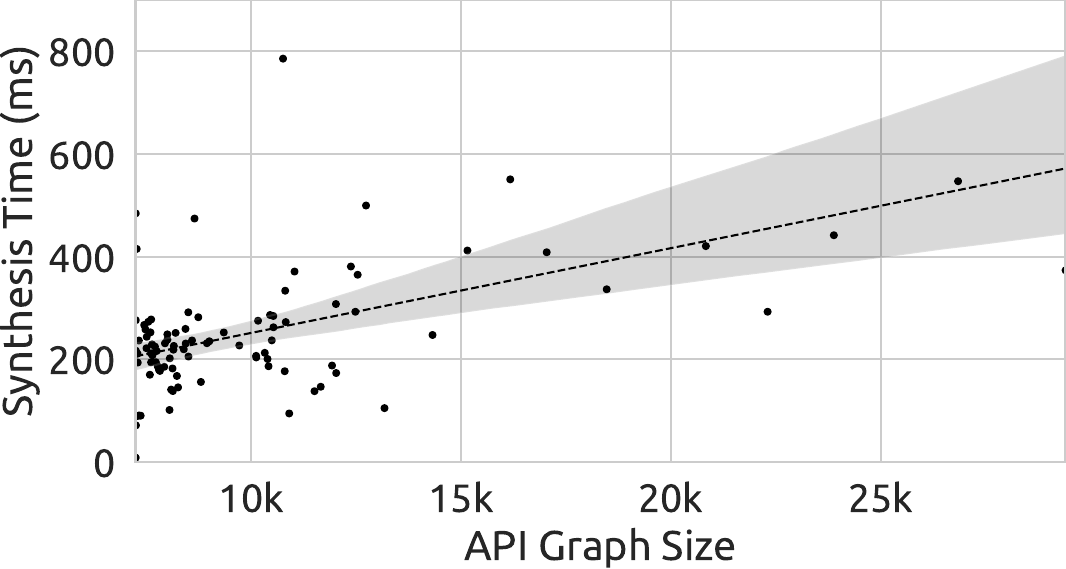}
        \caption{groovyc}
    \end{subfigure}%
    \hfill
    \begin{subfigure}{0.33\linewidth}
        \centering
        \includegraphics[width=\linewidth]{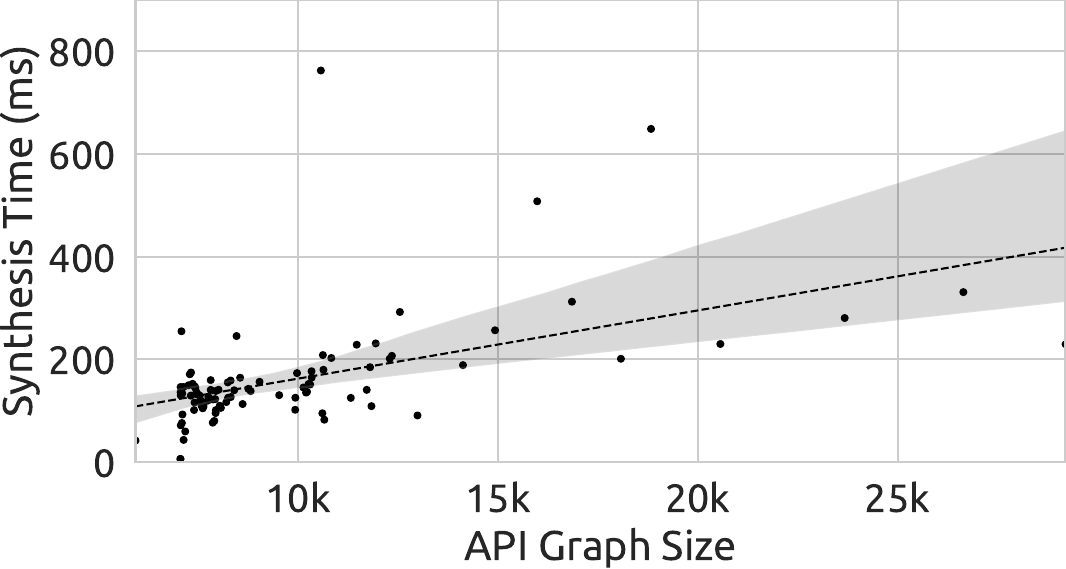}
        \caption{Dotty}
    \end{subfigure}%
    \hfill
    \begin{subfigure}{0.33\linewidth}
        \centering
        \includegraphics[width=\linewidth]{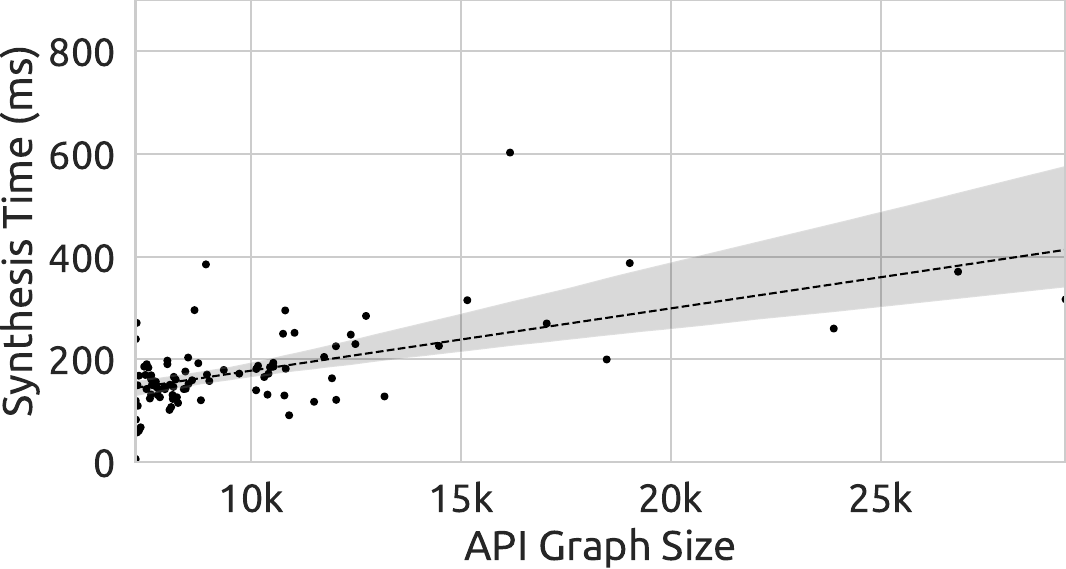}
        \caption{kotlinc}
    \end{subfigure}
    \vspace{-2mm}
    \caption{
        The synthesis time per program (in milliseconds)
        as a function of the size of the~\graph.
    }
    \label{fig:time}
    \vspace{2mm}
\end{figure*}
There are some exceptional
cases where the synthesis time
is relatively high
(see data points above the regression lines).
This is explained by
the high number of polymorphic definitions
with bounded type parameters
included in the corresponding APIs.
For example,
the~\href{https://github.com/assertj/assertj}{org.assertj:assertj-core}
library has a larger number of type constructors
and polymorphic methods that
define type parameters with recursive upper bounds
(e.g., {\tt T extends A<T>})
In such cases,
\tool~spends significant time
seeking appropriate type arguments 
to instantiate these polymorphic definitions.

\subsection{RQ4: Comparison of Thalia vs. Hephaestus}
\label{sec:eval-thalia-hephaestus}

\begin{figure*}[t!]
    \centering
    \begin{subfigure}{0.33\linewidth}
        \centering
        \includegraphics[width=\linewidth]{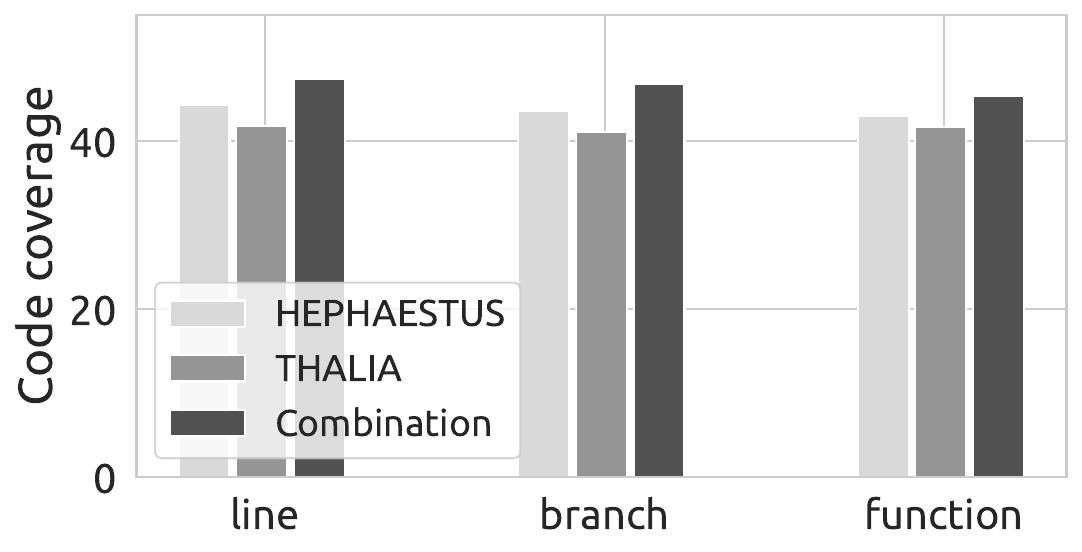}
        \caption{groovyc}
    \end{subfigure}%
    \hfill
    \begin{subfigure}{0.33\linewidth}
        \centering
        \includegraphics[width=\linewidth]{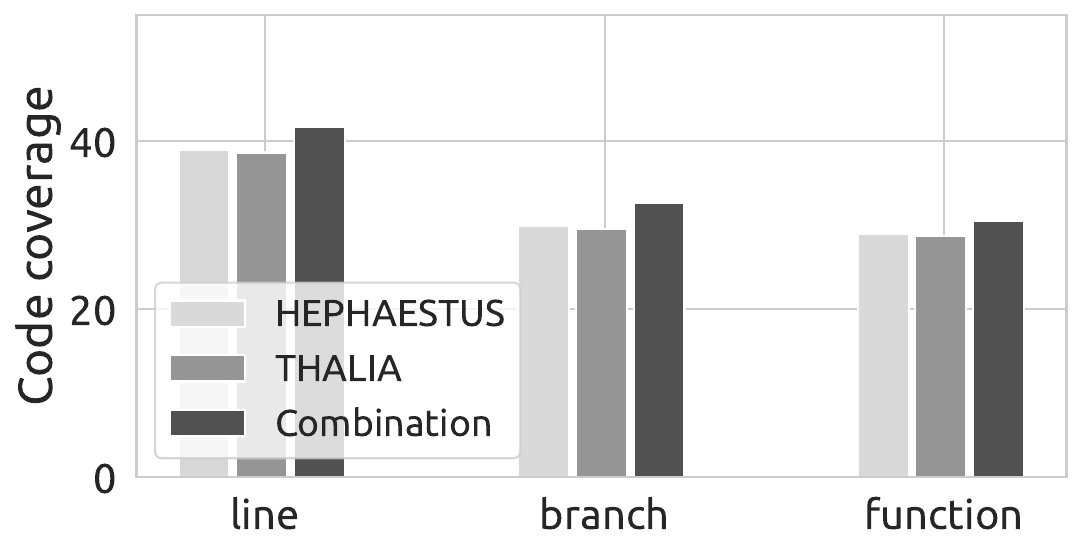}
        \caption{Dotty}
    \end{subfigure}%
    \hfill
    \begin{subfigure}{0.33\linewidth}
        \centering
        \includegraphics[width=\linewidth]{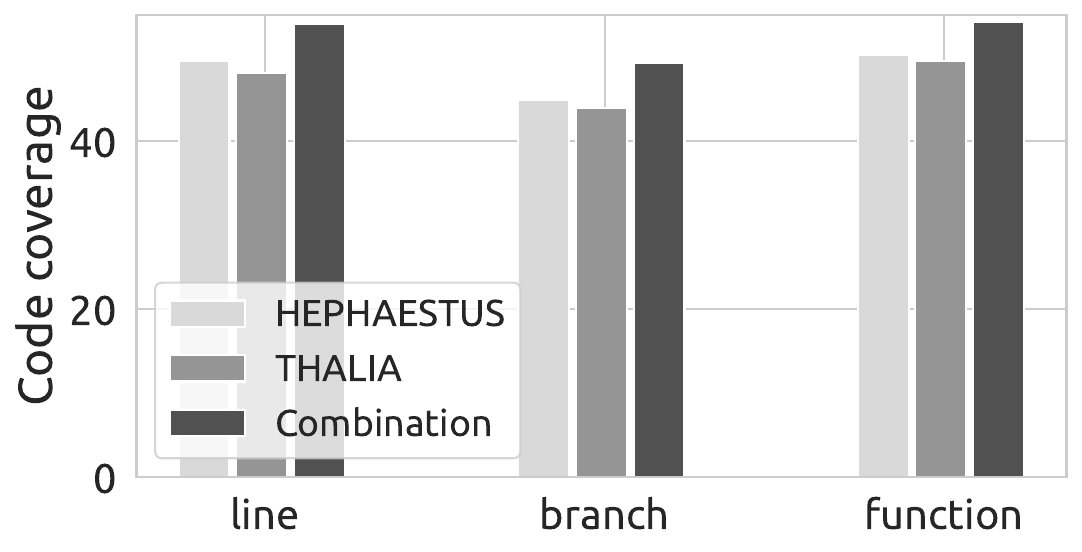}
        \caption{kotlinc}
    \end{subfigure}
    \vspace{-3mm}
    \caption{
        Comparison between code coverage of~\heph~and~\tool.
        }
    \label{fig:coverage-comparison}
    \vspace{1mm}
\end{figure*}

We compared our work with~\heph~\cite{hephaestus},
which is the state-of-the-art framework
for validating static typing implementations.
Originally,
\heph~had support for three languages,
two of which were also included in our evaluation
(Kotlin, Groovy).
To make~\heph~also generate Scala programs,
we followed the guidelines in the work~\citet{hephaestus}
and implemented a translator
that allows the conversion of programs
written in the~\heph 's IR into concrete Scala programs.

\point{Bug-finding capability}
We manually analyzed the bug-revealing test cases of~\tool~to
identify how many bugs are missed by~\heph.
Looking at Table~\ref{tab:features},
it is clear that the bug-revealing test cases exercise
a plethora of characteristics,
such as overloading and the use of inner classes,
which fall outside the capabilities of~\heph.
Specifically,
out of the~\ttotal bugs identified by~\tool,
\heph~fails to detect~\empirical{42} of them,
accounting for a~\empirical{51\%} miss rate.
Interestingly,
these~\empirical{42} bugs were found~\emph{effortlessly} by~\tool,
as it leverages existing compiled code
instead of explicitly supporting the generation
of intricate definitions.

For the remaining detected bugs, 
we also examined 
whether these bugs existed 
in the compiler versions 
used in the testing campaign by~\citet{hephaestus}. 
Remarkably, 
out of these bugs, 
\empirical{$9$} were found 
to exist in those versions, 
and~\heph~failed to detect them despite the thorough efforts.
This is explained by the expressiveness of
real-world APIs
that made~\tool~produce code
featuring complex combinations of typing features.
For example,
\href{https://issues.apache.org/jira/browse/GROOVY-11020}{GROOVY-11020} is triggered
by passing a function reference to a polymorphic method
that defines a type parameter bounded with a SAM type.
Although
\heph~supports polymorphic methods and SAM types,
there is an extremely low probability of generating
code with such a programming idiom.

Nevertheless,
there are certain bugs triggered by~\heph~that
cannot be detected by~\tool.
Most of these bugs have to do with
implementation flaws that are related to
the semantic validation of definitions.
Since \tool~employs pre-compiled API definitions,
it fails to detect these types of compiler errors.
For example,
\href{https://youtrack.jetbrains.com/issue/KT-49583}{KT-49583}
is a bug,
previously-reported by the team of~\heph.
The accompanying test case
contains an anonymous function
that encompasses another nested function declaration.
This issue is out of~\tool's reach,
because rather than producing new declarations
(e.g., a function in a lambda),
\tool~re-uses existing code.

\point{Code coverage analysis}
We conducted a $12$-hour run for both tools in every language, 
using all available modes 
and generating as many test cases as possible 
within that time frame. 
We ran~\heph~using the same settings
as described in the work of~\citet{hephaestus}.
We assessed the code coverage in each compiler 
using the programs generated by each tool. 

Figure~\ref{fig:coverage-comparison} illustrates the code coverage 
attained by~\heph~and~\tool.
\heph~slightly outperforms~\tool~on all the examined compilers. 
This outcome aligns with our expectations, 
as~\heph~produces larger programs 
and heavily emphasizes the utilization of declarations. 
However,
~\tool~offers the advantage of easily testing arbitrary features
(Section~\ref{sec:eval-test-case-chars}). 
Consequently, 
when combining the two approaches, 
we observe a percentage increase in \heph' line coverage of
$\empirical{7\%}$ for {\tt groovyc} and Dotty,
and~$\empirical{9\%}$ for {\tt kotlinc}.
This translates to $\nnum{\groovylinemissed}$--$\nnum{\kotlinlinemissed}$
additionally covered lines of code.
For branch coverage,
the combination of~\tool\ and~\heph\
leads to a~$\empirical{7\%}$ increase for {\tt groovyc},
$\empirical{9}\%$ for Dotty,
and~$\empirical{10\%}$ for {\tt kotlinc}.
This translates to $\nnum{\groovybranchmissed}$--$\nnum{\scalabranchmissed}$
additionally covered branches.
Finally,
when considering function coverage,
\tool\ contributes to
calling~$\empirical{5\%}$ more functions in {\tt groovyc} and Dotty,
and~$\empirical{8\%}$ more functions in {\tt kotlinc}.
This translates to $\nnum{\groovyfunctionmissed}$--$\nnum{\scalafunctionmissed}$
functions that were previously unexplored by~\heph.
We looked further into some packages and classes
\begin{wraptable}{r}{0.45\textwidth}
  \centering
  \vspace{-3mm}
  \caption{
      Average time of compiling and generating
      ~\heph~and~\tool~programs.
      Generation time is per program,
      whereas compilation time is per batch of 45 programs.
  }
  \vspace{-2mm}
  \label{tab:compilation-generation}
  \resizebox{1.0\linewidth}{!}{%
    \begin{tabular}{lrrrr}
    \toprule
    \multirow{2}{*}{Compiler} & \multicolumn{2}{c}{Compilation time (sec)} & \multicolumn{2}{c}{Generation time (sec)} \bigstrut \\
    & \tool & \heph & \tool & \heph \\
    \hline
    {\tt groovyc} & 3.1 & 6.6 & 0.3 & 0.7 \\
    Dotty         & 5.4 & 7.9 & 0.3 & 1 \\
    {\tt kotlinc} & 6.5 & 26.5 & 0.2 & 1.3 \\
    \bottomrule
    \end{tabular}
}
\end{wraptable}
to get an insight about that newly-explored code.
Some noteworthy examples include:
(1) the \texttt{org.jetbrains.kotlin.resolve.*} package,
which implements the name resolution procedure in {\tt kotlinc},
(2) a couple of classes  in Dotty
(e.g., {\tt dotty.tools.dotc.core.TypeComparer})
that handle polymorphic types and type comparisons.
As expected,
\tool~missed a lot of code associated with
declarations and types of expressions
(e.g., binary operators)
supported only by~\heph.
Overall,
our findings
indicate that~\tool~helps explore new code
in the compilers' code base,
addressing the current limitations of~\heph.


\point{Code size}
Within the 12-hour time period,
the average size of~\tool's programs
written in Groovy, Scala, and Kotlin 
is $\empirical{15}$, 
$\empirical{12}$, 
and $\empirical{11}$ LoC,
respectively.
\heph~generated programs
that are one order of magnitude larger,
with an average size of 
$\empirical{304}$, 
$\empirical{262}$,
and $\empirical{257}$ LoC for Groovy, Scala, and Kotlin respectively. 
The smaller size of~\tool's programs
simplifies the generation of minimal test-cases,
which are invaluable for compiler writers~\cite{c-reduce}.

\point{Generation and compilation time}
Table \ref{tab:compilation-generation} 
presents the average time spent by each tool
on generating and compiling programs. 
\heph~features higher times compared to~\tool, 
with varying performance across different languages. 
For example, 
{\tt groovyc} spends an average time of~$\empirical{6.6}$
seconds for compiling~45 \heph' programs,
while it needs only~$\empirical{3.1}$ seconds,
on average,
to compile 45 programs given by~\tool.
At the same time,
\tool~synthesizes a Groovy test case in 0.3 seconds,
while~\heph~requires~0.7 seconds to generate
a single Groovy program.
These findings highlight that~\tool~
can achieve greater throughput than the state-of-the-art work.

In short,
our comparative analysis
highlights four distinct
differences between
\tool\ and \heph.
(1) \tool's programs are considerably smaller
than those of~\heph\ (code size).
(2) \tool's programs involve
faster synthesis and compilation times
(generation and compilation time).
(3) Despite the reduced size of its test cases,
\tool\ still explores compiler regions
that~\heph\ does not reach
(code coverage analysis).
(4) Relying on existing compiled libraries makes
\tool\ effortlessly uncover~\empirical{51} bugs
missed by~\heph\ (bug-finding capability).

\subsection{Bug-Revealing Test Samples}
\label{eval:test-cases}

\begin{figure*}[t]
\centering 
\begin{subfigure}{0.48\textwidth}
\begin{lstlisting}[language=scala, basicstyle=\tiny\ttfamily]
def test() {
  val x: String = "strVal"
  Predef.identify[Function1[? >: Int, String]](x.substring)
}
// Definition of the API
class String {
  def substring(begin: Int): String
  def substring(begin: Int, end: Int): String
}
object Predef {
  def identity[A](x: A): A
}
\end{lstlisting}
\caption{\href{https://github.com/lampepfl/dotty/issues/17310}{DOTTY-17310}:
Fail to resolve the reference to method {\tt String.substring}.}
\label{fig:dotty-method-ref}
\end{subfigure}\hfil 
\vspace{1mm}
\begin{subfigure}{0.48\textwidth}
\begin{lstlisting}[language=kotlin, basicstyle=\tiny\ttfamily]
fun test() {
    val x: Iterable<HashSet<Number>> = TODO()
    val y: HashSet<Number> = TODO()
    // Type mismatch: inferred type is Number
    val res: List<HashSet<Number>> = x.minus(y)
}
// Definition of the API
package kotlin.collections;

operator fun <T> Iterable<T>.minus(element: T): List<T>
operator fun <T> Iterable<T>.minus(elements: Iterable<T>): List<T>
\end{lstlisting}
\caption{\href{https://youtrack.jetbrains.com/issue/KT-57596}{KT-57596}:
Resolving wrong method in the presence of operator overloading.}
\label{fig:kotlin-operator}
\end{subfigure}\hfil 
\begin{subfigure}{0.48\textwidth}
\begin{lstlisting}[language=Java, basicstyle=\tiny\ttfamily]
import com.google.common.collect.HashBasedTable;
 void test() {
    Number x = HashBasedTable.<Number, Number, Number>create().get(null, null);
  }
}
// Definition of the API
package com.google.common.collect;
interface Table<R,C,V> {
  default V get(Object x, Object y);
}
class HashBasedTable<R,C,V> implements Table<R,C,V> {
  static <R,C,V> HashBasedTable<R,C,V> create();
}
\end{lstlisting}
\caption{\href{https://issues.apache.org/jira/browse/GROOVY-11012}{GROOVY-11012}:
A bug in the treatment of default methods leads to an UTCE.}
\label{fig:groovy-bug-default}
\end{subfigure}
\hspace{1mm}
\begin{subfigure}{0.48\textwidth}
\begin{lstlisting}[language=scala, basicstyle=\tiny\ttfamily]
import java.util.Map
import org.apache.commons.lang3.Validate.notEmpty;
def test(): Unit = {
  val x: Map[String, String] = ???
  val res: Map[String, String] = notEmpty[Map[String, String]](x, "foo");
}
// Definition of the API
package org.apache.commons.lang3
object Validate {
  def notEmpty[T](x: Array[T]): T
  def notEmpty[T <: CharSequence](chars: T): T
  def notEmpty[T <: Map[?, ?]](map: T): T
}
\end{lstlisting}
\caption{\href{https://github.com/lampepfl/dotty/issues/17412}{DOTTY-17412}:
Dotty is unable to call an overloaded method with bounded type parameters.}
\label{fig:dotty-bounded-overload}
\end{subfigure}
\vspace{-2.0mm}
\caption{Sample test programs that trigger typing bugs.}
\vspace{2mm}
\end{figure*}


\point{Figure~\ref{fig:dotty-method-ref}}
This is a regression introduced in Scala 3.
The code creates a reference to 
an overloaded method of class {\tt String}
named {\tt substring}
and passes it as an argument
to the higher-order function {\tt identity}
which expects something of type {\tt Function1[? >: Int, String]}
(line 3).
Although the intention is
to make a reference to
the overloaded method
defined on line 7,
Dotty rejects the program with an error of the form:
\textit{``None of the overloaded alternatives of method substring match
expected type''}.
Interestingly,
when the expected type is {\tt Function1[Int, String]},
Dotty resolves the correct method as expected.

\point{Figure~\ref{fig:kotlin-operator}}
The code demonstrates an issue
in the design of Kotlin.
As with C++ and other languages,
Kotlin supports operator overloading.
The standard library of Kotlin
defines two overloaded methods for
the operator {\tt -} (minus).
When calling method {\tt minus},
{\tt kotlinc} decides to resolve the method
defined on line 11.
However,
this decision leads to a type mismatch,
as the inferred type of the method call becomes {\tt List<Number>},
while the expected type is {\tt List<HashSet<Number>{>}}.
The compiler should have resolved
the overloaded variant on line 10,
where the type variable {\tt T} is instantiated by the
type {\tt HashSet<Number>}.
This issue prevents the removal of polymorphic items
from mutable collections
(e.g., lists, sets)
via the corresponding operator,
i.e., {\tt x - y}.

\point{Figure~\ref{fig:groovy-bug-default}}
In this example,
the compiler erroneously rejects a well-typed program.
The root cause of this failure
lies in the way Groovy
treats interfaces with default methods
(see lines 8--10).
In Groovy,
an interface with a default method is implemented
as a trait,
a structural construct that allows
composition of behaviors
and implementation of interfaces.
{\tt groovyc} fails to
propagate the given type arguments
when calling {\tt create}
(line 3)
to the trait that implements the
interface {\tt Table} (lines 8, 12).
Consequently,
while checking the method call on line 3,
{\tt groovyc} infers the
return type of the method {\tt get} as {\tt Object}.
This leads to an UCTE,
as the expected type is {\tt Number}.

\point{Figure~\ref{fig:dotty-bounded-overload}}
This program calls an overloaded method
named {\tt notEmpty} that comes from
the {\tt org.apache.commons:commons-lang3} library (line 5).
The library contains a few polymorphic variants
of method {\tt notEmpty};
each of them defines a type parameter
with a unique upper bound
(lines 10--12).
The intention is to call the third overloaded alternative
(line 12),
as the explicit type argument of the method call on line 5
comes in line with the bound
of the underlying type parameter 
(i.e., {\tt Map[?,?]}).
Nevertheless,
an issue in Dotty's overload resolution procedure makes
the compiler reject the program by
reporting an ambiguous method call.
This issue stems from
the design choice of Dotty developers
to check bounds only after type checking
(to avoid cycles due to recursive upper bounds).
Therefore,
bounds do not influence overload resolution.
In this context,
it is impossible to properly call
the intended method from the given library.
Surprisingly,
when the definition on line 10 is removed,
the compiler behaves as expected.

\section{Related Work}
\label{sec:related}
\vspace{-2mm}

\point{Generative compiler testing}
{\it Generative compiler testing}
is an umbrella term for
methods that construct test programs
completely from scratch.
Csmith~\cite{csmith}
is the most influential program generator
which produces well-typed C programs while avoiding undefined behaviour.
Many subsequent generators have leveraged the power
of Csmith to
test other compilers (e.g., OpenCL)~\cite{many-core}
and compiler components~\cite{link-time},
or enhance the diversity of the generated programs~\cite{csmithedge,csmithedge2}.
A more recent program generator for C/C++
programs, YARPGen~\cite{yarpgen},
produces programs with 
complex arithmetic expressions
that are more likely to trigger optimization bugs.
Its re-implementation~\cite{yarpgen2} focuses on
validating loop optimizers.

A key challenge
associated with generative compiler testing
is the construction of code generators for syntactically-
and semantically-valid programs
that help test
beyond the front-end of a compiler.
This typically involves a large amount of
engineering effort,
while the corresponding implementation
is usually crafted
for a single target language.
Exploiting well-formed definitions
taken from existing software libraries
allows us to synthesize programs from scratch,
but at the same time,
easily port our program generator to multiple languages.

\point{Mutation-based compiler testing}
Mutation-based compiler testing
produces new test programs
by modifying existing ones.
The most effective compiler testing method
that lies in this category is
{\it equivalence modulo input (EMI)}~\cite{emi,deep-emi,live-emi}.
EMI profiles the execution of a seed program
and applies a set of transformations
so that the resulting programs
have the same output as the original one.
Applying semantics-preserving transformations
to existing programs
has been shown highly effective 
in the graphics driver domain~\cite{shader,spirv-fuzz}.

A primary limitation of mutation-based testing
is that its effectiveness is limited
to the quality of the available seed programs.
Our work tackles this limitation,
as mainstream languages come with a rich library ecosystem.
Library APIs expose a wide variety
of advanced features
that are more likely to trigger typing bugs~\cite{typing-study}.

\point{Program enumeration for compiler testing}
A number of {\it enumeration techniques}
have been developed for detecting issues
in deep compiler phases,
such as optimizations and code generation.
{\it Skeletal program enumeration (SPE)}~\cite{skeletal}
takes a program skeleton
and enumerates all program variants
that involve unique variable usage combinations.
SPE aims at exhibiting
diverse control- and data-dependence.
{\it Type-centric enumeration (TCE)}~\cite{kotlin-test}
finds crashes in the Kotlin compiler
by first generating a
reference program of a certain structure
and then replacing all program expressions
with other compatible expressions
of the same type.
In contrast to our work,
the outcome of TCE is not guaranteed to be type correct.

In this work,
we propose a different enumeration technique:
API enumeration focuses on
exposing diverse typing patterns
of API invocations to exercise 
interesting compiler behaviours
with regards to subtyping
and compile-time name resolution.

\point{Finding compiler typing bugs}
Most of the aforementioned techniques
target optimizing compilers.
The CLP-based ({\it Constraint Logic Programming})
program generator proposed by~\citet{rust}
and~\heph~\cite{hephaestus}
represent the work most closely related to our approach.
The CLP-based method works by
encoding the semantics rules of the language
under test into a set of logical constraints.
Querying a CLP engine under the given constraints
enables the generation of well-typed or ill-typed programs.
Despite its application to the Rust compiler,
CLP-based program generation comes
with many caveats that limit
its generality and practicality,
e.g., poor performances,
and even non-termination.

\heph\ employs {\it type graphs} with the primary
goal of capturing (1) the dependencies
between type variables
(i.e., determining whether a type variable can be inferred by another),
and (2) the inferred or declaration type of local variables.
A type graph is constructed using an intra-procedural analysis
on an existing program.
\heph\ consults type graphs
to maintain type correctness in its type erasure mutation
(Section~\ref{sec:type-erasure}).
In contrast,
our API graph captures
(1) all API definitions (methods, fields) found in an API,
(2) what types are needed to perform an application or a field access,
and (3) what is the type of each application or field access.
An API graph is constructed by traversing a given API specification,
rather than a program.
The main purpose of API graphs is to
identify type inhabitants.

Section~\ref{sec:eval-thalia-hephaestus}
thoroughly compares our work with~\heph.
In short,
\tool's underlying process,
which relies on API enumeration rather than
randomized program generation,
yields programs with distinct characteristics
(smaller code size)
and enables exercising compiler regions
that have been previously unexplored by \heph.

\point{Component-based program synthesis}
Our work is also related to
program synthesis techniques
that generate small code fragments
using components of existing libraries.
The purpose of these techniques
is not to test compilers,
but rather to assist developers
in programming tasks via library code reuse.
In this context,
a developer provides
an incomplete expression~\cite{partial-expressions,insynth}
or a method signature~\cite{rest-synthesis,petri-net,type-refinement}.
A program synthesis tool then produces
a ranked list of implementation sketches
that better match the developer's intent.

The idea of labelling~\graph s
with type substitutions shares
similarities with a graph data structure
called {\it equality-constrained tree automata (ECTA)}~\cite{ecta}.
In ECTA,
nodes are annotated with equality constraints.
While our~\graph\ constrains
type variable instantiations,
ECTA goes further by constraining arbitrary types,
such as matching an argument type with its corresponding
formal parameter type.
ECTA is more expressive than~\graph s,
because it solves the
more general program synthesis problem,
where the relevancy of the proposed solutions matters.
In contrast,
the primary use of~\graph s is much simpler:
we aim to identify chains of method calls/field accesses
of a certain type.
This allows for
(1) a compact representation,
especially the treatment of polymorphic types,
(2) an efficient enumeration done
by~\emph{standard} graph reachability algorithms
(e.g., Yen's algorithm),
and (3) identification of variable-length chains of method calls/field accesses,
in contrast to ECTA's fixed size approach.

\graph\ is inspired by {\sc prospector}~\cite{jungloid}.
{\sc prospector} introduces the {\it signature graph},
which treats every API component as a unary function
that takes an input type (receiver type),
and produces an output type (return type).
Contrary to our~\graph,
{\sc prospector} does not handle parametric polymorphism
which is a key language feature for revealing
typing bugs~\cite{typing-study}.
To handle parametric polymorphism,
we enrich {\sc prospector}
with Algorithm~\ref{alg:find-paths}
and type substitution labels.
A generalization of {\sc prospector}
is {\sc SyPet}~\cite{petri-net},
which models the structure of an API
using a Petri-net.
\citet{type-refinement} extend {\sc SyPet}
by introducing the {\it type-guided abstraction refinement}
that allows program synthesis over polymorphic types
using an SMT encoding.
Finally,
InSynth~\cite{insynth} assigns weights to both API definitions and types,
modeling program synthesis as an optimization problem.
InSyth leverages these weights
to guide its search for type inhabitants,
selecting expressions that minimize a specific weight function.

The goal of these techniques is different
from ours.
They strive to synthesize the most optimal solution
e.g., in terms of code size,
number of API method invocations,
type distance, or
users' intent.
Also many of these tools
lack support for parametric polymorphism,
or exhibit high running times.
Such challenges can make
these tools less suitable
for our compiler testing scenarios.

SyRust~\cite{rust-lib} is a semantic-aware program
synthesis technique for testing Rust libraries.
Using a manually-generated code template and inputs,
it encodes the typing rules of Rust 
(w.r.t. ownership, variable lifetime)
into a satisfiability problem
and synthesizes test cases of increasing size.
Contrary to our work,
SyRust is semi-automatic and faces scalability issues,
because it is capable of handling
15 API definitions (e.g., methods) per library.
Also,
SyRust's focus lies in identifying
memory-safety bugs in library implementations,
rather than producing test cases
that showcase diverse typing patterns
(e.g., type inference reasoning)
to test compilers.

\section{Conclusion}
\label{sec:conclusion}

We have presented
an API-driven program synthesis approach
for testing the implementation of compilers'
static typing procedures.
Our method harnesses the ubiquity and complexity
of APIs extracted from established software libraries,
and synthesizes concise client programs
that employ API entities
(e.g., types, functions, fields)
through different typing patterns.
This helps us exercise a broad spectrum of
type-related compiler functionalities,
without the need to generate
complex API definitions ourselves.
Our evaluation on the compilers
of Scala, Kotlin, and Groovy has shown the
effectiveness of our approach.
Indeed, our implementation has uncovered $\ttotal$ bugs
($\treal$ bugs are either confirmed or fixed),
$\empirical{51}$ of which could not have been detected by prior work.

Our API-driven program synthesis approach
opens up further opportunities
for combining synthetic and real-world code
to discover bugs in compiler implementations.
For example,
in addition to generating client programs,
we can extend our approach to create intricate
definitions (e.g., classes)
that build upon existing ones found in an API
(e.g., through inheritance constructs).
We believe that
this extension will enable
the identification of new compiler bugs.


\bibliographystyle{ACM-Reference-Format}
\bibliography{main}

\end{document}